\documentclass[journal,final,a4paper,twocolumn,10pt]{IEEEtran}%
\pdfoutput=1

\usepackage{etex}

\usepackage[english]{babel}
\usepackage{makeidx}
\makeindex

\usepackage[absolute,overlay]{textpos} %
\usepackage[usenames,svgnames,table]{xcolor} %

\usepackage{algorithm}
\usepackage{algorithmicx}
\usepackage{algcompatible}
\usepackage{longtable}

\usepackage[normalem]{ulem} %
\usepackage{pifont} %
\usepackage{amsmath} %
\usepackage{amsfonts} %
\usepackage{amssymb}
\usepackage{array} %
\usepackage{graphicx} %

\usepackage{xifthen}

\usepackage[colorlinks=false,pdfborder={0 0 0},plainpages=false]{hyperref}

\usepackage[english]{babel} %
\usepackage{hyphenat} %
\usepackage{setspace}
\usepackage{url}
\usepackage{framed}
\usepackage{fancyhdr}
\usepackage{lastpage} %
\usepackage[us,24hr]{datetime} %

\usepackage{rotating} %
\usepackage{dcolumn} %
\usepackage{pdfpages} %
\usepackage{setspace}
\usepackage{amssymb} %
\usepackage{multirow} %
\usepackage[labelfont=bf,textfont=bf]{caption} %
\usepackage{tabularx}
\usepackage{tabulary}
\usepackage{colortbl}
\usepackage{balance}
\usepackage{placeins} %
\usepackage{multicol} %
\usepackage{lipsum} %

\usepackage{fancyvrb}
\usepackage{enumitem} %
\usepackage{parskip}
\usepackage{needspace} %
\usepackage[labelformat=simple]{subcaption}

\usepackage[utf8]{inputenc}
\usepackage[T1]{fontenc}

\textblockorigin{0pt}{0pt}

\makeatletter

\newcommand{\facp}{\@ifstar \facpStar \facpNoStar}

\newcommand{\facpNoStar}[1]{%
\expandafter\ifx\csname ac@#1\endcsname\AC@used %
\acs*{#1s}\acused{#1}%
\else
\acf*{#1s}\acused{#1}%
\fi
}

\newcommand{\facpStar}[1]{%
\expandafter\ifx\csname ac@#1\endcsname\AC@used %
\acs*{#1s}%
\else
\acf*{#1s}%
\fi
}

\makeatother

\makeatletter
\newcommand{\checkhere}[1]{%
\leavevmode\vbox to\z@{%
\vss%
\rlap{\vrule\raise .75em%
\hbox{\underbar{\normalfont\footnotesize\ttfamily \ding{43} #1}}}}%
\@latex@warning{A "checkhere" mark is still present and uncommented.}%
}
\makeatother

\makeatletter
\newcommand{\hg}[1]{%
\noindent\ding{43} \texttt{\footnotesize\uline{#1}} \hfil \\ %
\@latex@warning{A suggestion is still present and uncommented.}%
}
\makeatother

\newcolumntype{d}[1]{D{.}{.}{#1}}
\newcolumntype{.}{D{.}{.}{-1}}
\makeatletter
\newcolumntype{E}[1]{>{\boldmath\color{red}\DC@{.}{.}{#1}}c<{\DC@end}}
\newcolumntype{F}[1]{>{\boldmath\color{green}\DC@{.}{.}{#1}}c<{\DC@end}}
\makeatother

\definecolor{good}{rgb}{0,.8,0}
\definecolor{bad}{rgb}{.9,0,0}
\definecolor{mygreen}{rgb}{0.6,0.8,0}
\definecolor{myorange}{rgb}{1,0.4,0}
\definecolor{myblue}{rgb}{0.2,0.4,1}
\definecolor{myviolet}{rgb}{0.498,0,0.498}
\definecolor{myred}{rgb}{0,0.498,0.498}

\RequirePackage[normalem]{ulem}
\RequirePackage{color}\definecolor{RED}{rgb}{1,0,0}\definecolor{BLUE}{rgb}{0,0,1}

\definecolor{quote}{rgb}{.13,.33,.0}

\makeatletter
\def\FrameCommand{%
\pdfcolorstack\@pdfcolorstack push{\current@color}%
\hspace{.3in} \vrule width 3pt%
\pdfcolorstack\@pdfcolorstack pop\relax%
\fboxrule=\FrameRule \fboxsep=\FrameSep \fbox%
}
\makeatother

\newcommand{\modification}[1]{%
\MakeFramed {\advance\hsize -1.3in \FrameRestore}%
\noindent%
#1
\endMakeFramed
}

\makeatletter

\usepackage{accents}
\usepackage{pgfplots} 
\pgfplotsset{compat=1.9}

\graphicspath{{./figures/}{./figures/png/}}

\usepackage{ifdraft}

\ifdraft{\renewenvironment{multline}{\begin{equation}}{%
    \end{equation}\ignorespacesafterend%
}}{}

\newtheorem{algo}{\textbf{\textit{Algorithm}}}
\newtheorem{lemma}{\textbf{\textit{Lemma}}}
\newtheorem{theorem}{\textbf{\textit{Theorem}}}
\newtheorem{corollary}{\textbf{\textit{Corollary}}}

\newtheorem{definition}{\textbf{\textit{Definition}}}
\newtheorem{conjecture}{\textbf{\textit{Conjecture}}}

\newenvironment{proof}{\textit{Proof:}}{\hfill$\blacksquare$}

\begin{document}

\date{\today}

\title{Lower Bounds on the Redundancy of Huffman Codes with Known and Unknown Probabilities}

\author{Ian~Blanes, %
		Miguel Hernández-Cabronero,\\
        Joan~Serra-Sagrist{\`a}, %
		and~Michael W. Marcellin%
        
\thanks{This work was supported in part by
the Centre National d'Etudes Spatiales,
and by the Spanish Ministry of Economy and Competitiveness (MINECO), by
the European Regional Development Fund (FEDER) and by the European Union
under grant TIN2015-71126-R and RTI2018-095287-B-I00, and by the Catalan Government
under grants 2014SGR-691 and 2017SGR-463.}

}

\markboth{}{}

\graphicspath{ {../figures/} }

\maketitle

\begin{abstract}
\noindent 
In this paper we provide a method 
to obtain tight lower bounds on the minimum redundancy achievable by a Huffman code 
when the probability distribution underlying an alphabet is only partially known.
In particular, we address the case where the occurrence probabilities are unknown for some of the symbols in an alphabet. Bounds can be obtained for alphabets of a given size,
for alphabets of up to a given size, and for alphabets of arbitrary size. 
The method operates on a Computer Algebra System, yielding closed-form
numbers for all results.
Finally, we show the potential of the proposed method to shed some light on
the structure of the minimum redundancy achievable by the Huffman code.
\end{abstract}


\section{Introduction} \label{sect1}

\newcommand{\p}{\text{\rm p}}
\newcommand{\len}{\text{\rm l}}

\IEEEPARstart{L}{et} $A$ be a discrete memoryless source of $n \ge 2$ symbols,
having an alphabet $\{a_1,\dots,a_n\}$ with associated probability distribution
$\{\p(a_1),\dots,\p(a_n)\}$.
Let
\begin{equation}
\mathcal{R}_C(A)=\mathcal{L}_C(A) - \mathcal{H}(A)
\end{equation}
be the redundancy achieved by a code $C$, where 
\begin{equation}
\mathcal{L}_C(A)=\sum_{i=1}^n \p(a_i) \cdot \len_C(a_i)
\end{equation}
is the average length of $C$ when applied to source $A$, $\len_C(a_i)$ is the length of the codeword for $a_i$, and

\begin{equation}
\mathcal{H}(A)=-\sum_{i=1}^n \p(a_i)\cdot \log \p(a_i)
\end{equation}
is the entropy of $A$. Without loss of generality, we assume that $\p(a_i) > 0$ for all $i$ and that all logarithms are base $2$.

The Huffman Algorithm~\cite{H52} can be employed to produce 
a Huffman code, $H$, for a source $A$. Huffman codes are prefix-free and are
optimal in the sense that no other code can achieve
a lower redundancy than $\mathcal{R}_H(A)$ when coding one symbol at a time.

Producing bounds on the redundancy of a Huffman code  
for which the underlying probability distribution is only partially known
is a recurring theme in the literature.
Many authors have described successively more accurate bounds on the redundancy of a Huffman code when only the probability of the most-likely symbol is known. 
In~\cite{Gal78}, an upper bound is provided for this case. 
An improved upper bound and a lower bound are described in~\cite{Joh80}.
In particular, a tight lower bound is provided when the (known) probability of the most-likely symbol is greater than $0.4$.
The previous upper bound is later extended for the D-ary case in~\cite{GO87}.
Bounds are further refined in~\cite{CGT86}, and~\cite{MA87} provides a lower bound that is tight for all values of the probability of the most likely symbol.
Improvements on the upper bound 
are provided in~\cite{CD89}.
In~\cite{Bir90}, the lower bound provided in~\cite{MA87}
is improved for a known alphabet length. 

The concept of local redundancy, a useful tool to prove previous results through different means, is introduced in \cite{Yeu91}.
Further bound improvements are provided in~\cite{CD91,Man92,PD96,DD97,Bae11}.
Tight upper and lower bounds are provided for one known symbol (not necessarily the most probable symbol) in~\cite{MPK06}.
The lower bound is obtained through convex optimization
over a small set of prefix-free codes sharing a singular structure.
Interestingly, the tight lower bound for one known symbol is the same as the one obtained in~\cite{MA87} for the known most-likely symbol case.
In~\cite{Stu94}, redundancy bounds are provided for binomially-distributed source words. 
Extending the previous work, convergent and oscillatory asymptotic behaviors 
are described as word sizes become larger in~\cite{Szp00}.
Application to Shannon codes with Markov inputs is futher described in~\cite{MS13}.

The main contribution of this paper is an algorithm (presented in Section~\ref{sect4}) which obtains a tight lower bound on the minimum redundancy achievable by a Huffman code for sources in which occurrence probabilities are only known for some
symbols. In contrast with previous bounds, the proposed method 
obtains bounds, in general, for two or more known probabilities.
The algorithm can be executed in a Computer Algebra System~(CAS)
and produces its result as a closed-form number~(as defined in~\cite{Cho99}).

While this result sheds light on the 
structure of Huffman code redundancy,
our main motivation behind pursuing this problem is to further advance 
the still challenging problem of finding optimal variable-to-variable (V2V)
codes~\cite{Kho72}.
Twenty-six years ago, in one of the earliest attempts at V2V code construction, Fabris already hinted at ``no algorithmic solution better than exhaustive search''~\cite{Fab92}.
Later, non-optimal construction methods were provided in~\cite{SDJ04}
and in~\cite{BDS08}, where low-redundancy (and even vanishing-redundancy) codes are achieved at the expense of increased code tables.
Fortunately, V2V codes with precomputed code tables of manageable size can be used in practice, as illustrated for example by their use in~\cite{CCSDS-MHDC-123B2-BlueBook},
where they are employed as a higher-efficiency replacement for Golomb codes~\cite{Gol66}.
A similar approach could be followed in recent compression schemes that use Golomb codes, such as~\cite{CLS+16,MRW17,TK18}.

While a feasible method for medium-sized optimal V2V code construction has not yet been found,
exhaustive search may still prove to be a useful approach.
Recently, Kirchhoffer et al. explore exhaustive search and provide strategies to reduce the complexity of a search for minimum redundancy~\cite{KMS+18}.
The work presented herein can be employed to significantly prune the search space of V2V codes.
Even if we do not directly address V2V codes in this manuscript, they are an important motive
for this research, and thus a few pruning examples will be provided in Section~\ref{sect5}.

Another significant contribution of this work is the presentation of novel redundancy maps in Section~\ref{sect5},
which provide a visual representation of minimum redundancy and insight into which code covers each minimum redundancy region. 
In particular, from one of the redundancy maps we can extend~\cite{MPK06} through
a conjecture for the minimum redundancy formula for the case of two known probabilities.
Moreover, the framework provided herein has potential for further interesting extensions, such as the imposition of
more general constraints on probability distributions. For example, in addition to the equality constraints discussed in this work, the method could be extended to include inequality constraints,  e.g., at least one probability is less than $0.1$ and at least one probability is greater than $0.5$.

\subsection{The Huffman Algorithm}

For the purpose of establishing a common notation
and being able to draw parallels, a description of the Huffman Algorithm follows for the (usual) case where all symbol probabilities of a source are known.
For simplicity, and with some abuse of notation, we write
$A=\{a_1,\dots,a_n\}$ to describe a source $A$ with alphabet $\{a_1,\dots,a_n\}$, and associated occurrence probabilities $\{\p(a_1),\dots,\p(a_n)\}$. The number of symbols in the alphabet is denoted by $|A|$.

The Huffman Algorithm is described here in terms of a state machine. 
In the following, the current state $\Theta^{(i)}$ is successively updated by a state transition function $\mathrm{h}(\cdot)$.
At this point in the development, the state takes the form of a discrete source.
Specifically, given a source $A=\{a_1,\dots,a_n\}$, state transition function $\mathrm{h}(\cdot)$ results in a new source by merging symbols $a_j$ and $a_k$ into a new symbol, denoted by $\left[a_j,a_k\right]$ with associated probability 
$\p\left(\left[a_j,a_k\right]\right) = \p\left(a_j\right) + \p\left(a_k\right)$, as follows:
\begin{equation}
\mathrm{h}\left(A,j,k\right) = \big( A
\setminus \left\{ a_j, a_k \right\} \big) \cup \{ \left[a_j,a_k\right] \}.
\end{equation}

\begin{algo}[The Huffman Algorithm] 
Let $A=\{a_1,\dots,a_n\}$ be a discrete source of $n\ge 2$ symbols.
Let $\Theta^{(i)} = \{\theta^{(i)}_1,\dots,\theta^{(i)}_{n-i}\}$
be the state in the $i$th iteration.

\begin{enumerate}
\item Let $\Theta^{(0)} = A$ and set $i \gets 0$.
\item Find indices $j$ and $k$, with $1 \le j<k \le n - i$, of two elements in $\Theta^{(i)}$ such that no other
element in $\Theta^{(i)}$ has a smaller occurrence probability.
 I.e.,
\[ \p(\theta_j^{(i)}) \le \p(\theta_l^{(i)}) \hspace{1em} \forall\, l \text{ s.t. } l \neq k, \]
\[ \p(\theta_k^{(i)}) \le \p(\theta_l^{(i)}) \hspace{1em} \forall\, l \text{ s.t. } l \neq j. \]

\item Merge the two elements to obtain $\Theta^{(i+1)} = \mathrm{h}(\Theta^{(i)}, j , k)$.
\item Set $i \gets i + 1$.
\item If $i < n - 1$, go to step 2.
\item Stop.
\end{enumerate}
\label{a1}
\end{algo}

The result of this algorithm is a source $\Theta^{(n-1)}$ containing a single symbol 
$\theta_1^{(n-1)}$ composed by the repeated merging of the original symbols. This symbol is equivalent to the well-known representation of a Huffman code in tree form. 
The codeword lengths $\len(a_i)$ for the original symbols $a_i$ can be obtained from $\theta_1^{(n-1)}$ by counting the number of ``$\,\,]\,$'' minus the number of ``$\,[\,\,$'' to the right of $a_i$.
For example, for the case of $A=\{a_1,a_2,a_3,a_4,a_5\}$ with 
$\p(a_1)=0.10$,
$\p(a_2)=0.21$,
$\p(a_3)=0.15$,
$\p(a_4)=0.30$, and
$\p(a_5)=0.24$, the Huffman Algorithm yields (within a permutation)
\begin{equation}
\theta_1^{(4)} = \Big[[a_2,a_5],\big[a_4,[a_1,a_3] \big] \Big].
\end{equation}
The resulting codeword lengths are 
$\len(a_2) = \len(a_5) = 4 - 2 = 2$,
$\len(a_4) = 3 - 1 = 2$, and
$\len(a_1) = \len(a_3) = 3 - 0 = 3$.

A similar scan can be applied to obtain codewords.
The codeword for symbol $a_i$ can be produced, in reverse order, by 
scanning $\theta_1^{(n-1)}$ starting at $a_i$ and proceeding to the right.
Occurrences of other alphabet symbols are ignored. Also, 
whenever a ``$\,[\,\,$'' is found, anything up to (and including) the matching ``$\,\,]\,$'' is ignored.
A zero is produced for each ``$\,\,]\,$'' found when immediately preceded by ``$\,,\,$''
and a one otherwise.
This scan yields
$110$,
$00$,
$111$,
$10$, and
$01$ for $a_1$ to $a_5$, respectively.

It is worth noting that the condition $j<k$ in step~2 avoids generating certain codes that are equivalent within a permutation to the codes produced by the algorithm. This facilitates
complexity reductions as discussed in future sections.

To apply Algorithm~\ref{a1} to a source $A$, the underlying probability distribution must be fully specified, which is contrary to the objective of this work --- operating with only partially-known probability distributions.
In Section~\ref{sect2}, we address the issue of partially-known probability distributions 
for alphabets of given size $n$. In Section~\ref{sect3}, we consider the case of alphabets of size up to $n$, and provide a general bound for arbitrary alphabet sizes. In 
Section~\ref{sect4} we provide an efficient method to calculate the general bound,
and in Section~\ref{sect5} some examples are shown. Finally,
conclusions are drawn in Section~\ref{sect6}.

\section{A Bound on Huffman Code Redundancy} \label{sect2}

Suppose now that the underlying probability distribution of a discrete memoryless source $A$ of $n\ge 2$ symbols has $m$ symbols with known probabilities, $0 \le m \le n$, whereas the probability is unknown for the remaining $n-m$ symbols, $0 \le n-m \le n$.%
\footnote{When $m=n$, there are no unknown probabilities, and the exact redundancy can be calculated directly from~\eqref{a1}. Similarly, when $m=n-1$, it is possible to compute the one ``unknown'' probability (since probabilities must sum to $1$), 
and again, the exact redundancy can be computed via~\eqref{a1}.
Nevertheless, the bounds proposed herein can be computed in both cases, and yield the redundancy that would be obtained by~\eqref{a1}.}
Notation and terminology for these partially-known sources are established in the following definitions.

\begin{definition}[Sub-source] 
Let $A$ be a source with alphabet $\{a_1,\dots,a_n\}$ and associated probabilities $\p(a_i)$. We define a sub-source of $A$ as a subset of the symbols of the alphabet of $A$ together with their associated probabilities $\p(a_i)$. 
If $X=\{x_1,\dots,x_m\}$, $0\le m\le n$, is a sub-source of $A$, we write $X \sqsubseteq A$. The cardinality $m$ of the sub-source is denoted by $|X|$. We note that the probabilities of a sub-source do not necessarily add up to $1$.
\end{definition}

\begin{definition} [Complementary Sub-source] 
Given a sub-source $X$ of a source $A$, the complementary sub-source $Y$ of $X$ with respect to $A$ is the sub-source containing all the symbols of $A$ not in $X$ together with their associated probabilities.
\end{definition}

A source $A$ for which some symbol probabilities are known can be thought of
as two complementary sub-sources: a sub-source $X$ holding all symbols for which probabilities
are known, and a complementary sub-source $Y$ holding all symbols for which probabilities
are unknown.

In the remainder of this section we describe how to obtain a bound on the lowest possible redundancy obtainable by a Huffman code for a source $A$ of $n\ge 2$ symbols, with $m$ of the probabilities being known, $0\le m\le n$. We do this by considering a sub-source $X$ of $A$, with $|X|=m$. We then find the lowest possible redundancy obtainable by a Huffman code over all sources $B$ such that $X \sqsubseteq B$
and $|B| = |A| = n$. We formalize this bound in the definition below.

\begin{definition} [Redundancy Bound for Sources of $n$ Symbols] 
Let $X$ be a sub-source of $m$ symbols for a source $A$ of $n$ symbols, with $0\le m \le n$ and $n \ge 2$. Then
\begin{equation}
\mathcal{R}^{(n)}_\text{min}(X) = \min_{ \big\{ B\, \bigm\vert X \sqsubseteq B,\, |B| = n \big\} } 
\big\{\mathcal{R}_H(B) \big\}.
\label{eq:rmin_def}
\end{equation}
\end{definition}

Clearly, $\mathcal{R}^{(n)}_\text{min}(X)$ is a lower bound to the redundancy obtainable by a Huffman code for $A$. In the following we describe a method to compute $\mathcal{R}^{(n)}_\text{min}(X)$.

\begin{theorem}\label{t1}
$\mathcal{R}^{(n)}_\text{min}(X)$ can be obtained as
\begin{equation}
\mathcal{R}^{(n)}_\text{min}(X) =
\min_{C \in \Phi^{(n)}} 
\Big\{
\min_{ \big\{ B\, \bigm\vert X \sqsubseteq B,\, |B| = n \big\} } 
\big\{
\mathcal{R}_C(B) \big\} \Big\},
\label{eq:rmin_calc}
\end{equation}
where $\Phi^{(n)}$ is the set of all possible Huffman codes  
that can be generated by Algorithm~\ref{a1} for sources of $n$ symbols.
\label{t:separation}
\end{theorem}

\begin{proof}
As Algorithm~\ref{a1} always yields an optimal code~\cite{H52},
a Huffman code is included in $\Phi^{(n)}$ for every $B$, $|B|=n$. Otherwise,
there would be a source $B$ for which Algorithm~\ref{a1} does not yield the optimal code.
Hence, 
\begin{equation}
\mathcal{R}_H(B)  = \min_{C \in \Phi^{(n)}} \big\{ \mathcal{R}_C(B) \big\},
\end{equation}
which substituted into \eqref{eq:rmin_def} yields
\begin{equation}
\mathcal{R}^{(n)}_\text{min}(X) = \min_{ \big\{ B\, \bigm\vert X \sqsubseteq B,\, |B| = n \big\} } 
\Big\{
\min_{C \in \Phi^{(n)}} \big\{
\mathcal{R}_C(B) \big\} \Big\}.
\end{equation}
The minimization operations can be permuted, which yields~\eqref{eq:rmin_calc}.
\end{proof}

While its proof is simple,
the implication of Theorem~\ref{t1} is significant. In particular, it allows the problem of
calculating $\mathcal{R}^{(n)}_\text{min}(X)$ to be decomposed into the following two problems which can be solved sequentially:
\begin{itemize}
\item[(a)]
the problem of obtaining $\Phi^{(n)}$, and
\item[(b)]
the problem of obtaining the lowest redundancy 
achievable for each $C \in \Phi^{(n)}$:
\begin{equation}
\min_{ \big\{ B\, \bigm\vert X \sqsubseteq B,\, |B| = n \big\} } \big\{ \mathcal{R}_C(B) \big\}.
\label{eq10}
\end{equation}
\end{itemize}

We address problem (a) in Subsection~\ref{phi_n}, where we show how to obtain $\Phi^{(n)}$ through exhaustive enumeration, obtaining feasible code structures while implicitly enforcing integer codeword lengths. As for problem (b), we show in Subsection~\ref{sec:minimization} that it can be posed as a constrained convex optimization problem and solved accordingly. In subsequent sections we show how to obtain a dramatically smaller, yet still sufficient, subset of $\Phi^{(n)}$.

\begin{corollary}
The right-hand side of \eqref{eq:rmin_calc} is a tight lower bound to the redundancy obtainable by a Huffman code for $A$ given $X$.
\end{corollary}
\begin{proof}
For the code $C\in \Phi^{(n)}$ that minimizes the outer minimization operation
in~\eqref{eq:rmin_calc}, 
solving the convex optimization problem in (b) yields both 
\eqref{eq10}
and the values of $\p(b_i)$.
Thus there exists a source $B=\{b_1,\dots,b_n\}$ for which $\mathcal{R}_C(B)$ lies on the bound.
\end{proof}

\newcommand{\ua}{\mathcal{U}(A)}
\newcommand{\ub}{\mathcal{U}(B)}
\newcommand{\ka}{\mathcal{K}(A)}
\newcommand{\kb}{\mathcal{K}(B)}

\subsection{All Huffman Codes of $n$ Codewords}\label{phi_n}

The following algorithm generates a set of prefix-free codes of $n$ codewords
which includes all possible Huffman codes that can be produced by Algorithm~\ref{a1}. 
The algorithm explores every possible state trajectory $\Theta^{(0)},\ldots,\Theta^{(n-1)}$.
It does so by iterating over a set $\Phi^{(i)}$ containing all possible states $\Theta^{(i)}$
at iteration $i$.

\begin{algo} 
Let $A=\{a_1,\dots,a_n\}$ be an alphabet of $n\ge 2$ symbols, and let
$\Phi^{(i)}$ be a set of states.

\begin{enumerate}
\item Let $\Phi^{(1)} = \{ A \}$ and set $i \gets 1$.
\item Produce $\Phi^{(i+1)}$ by considering all possible 
applications of the state transition function $\mathrm{h}(\cdot)$
to all states in~$\Phi^{(i)}$. I.e.,
$$\Phi^{(i+1)} = \left\{\mathrm{h}(\Theta, j, k) 
\mid \Theta \in \Phi^{(i)},\,\,
 1 \le j < k \le |\Theta| %
  \right\}.$$
\item Set $i \gets i + 1$.
\item If $i < n$, go to step 2.
\item Stop.
\end{enumerate}
\label{a2}
\end{algo}

The result of this algorithm is $\Phi^{(n)}$, which is a set of states. 
Each such state is an alphabet containing a single symbol $\theta$ composed by the repeated merging of the original symbols. As  
described previously with respect to Algorithm~\ref{a1},
each such~$\theta$
is equivalent to a prefix-free code in the form of a tree. 

For example, for the case of $A=\{a_1,a_2,a_3\}$, 
\begin{equation}
\Phi^{(1)} = \big\{ \{a_1,a_2,a_3\} \big\},
\end{equation}
\begin{equation}
\Phi^{(2)} = \big\{ 
\{a_3,[a_1,a_2]\},
\{a_2, [a_1,a_3]\},
\{a_1, [a_2,a_3]\}
\big\},
\end{equation}
and
\begin{equation}
\Phi^{(3)} = \big\{ 
\{[a_3,[a_1,a_2]]\},
\{[a_2, [a_1,a_3]]\},
\{[a_1, [a_2,a_3]]\}
\big\}.
\end{equation}

Examination of $\Phi^{(3)}$ yields three different prefix-free codes.
Each code has two codewords of length~$2$ and one of length~$1$.

We note that Algorithm 2, in fact, generates the set of all prefix-free codes for alphabets of $n$ symbols (discounting for equivalent permutations).
It also introduces the concept of state transitions which is used extensively below.

With some abuse of notation, we employ $\Phi^{(n)}$ to denote the set of associated prefix-free codes employed in the outer minimization of~\eqref{eq:rmin_calc} in Theorem~\ref{t:separation}. In the following subsection, we tackle the inner minimization.

\subsection{Convex Optimization}\label{sec:minimization}

Given a sub-source $X$ and a code $C \in \Phi^{(n)}$,  we denote the inner minimization of~\eqref{eq:rmin_calc} by $\mathcal{F}(X,C)$. That is,
\begin{equation}
\mathcal{F}(X,C)=\min_{ \big\{ B\, \bigm\vert X \sqsubseteq B,\, |B| = n \big\} } \big\{ \mathcal{R}_C(B) \big\}.
\label{eq:14}
\end{equation}
Simply put, for a given prefix-free code $C$ having known codeword lengths $0<\len_C(b_i)<\infty$, and a sub-source $X$ having known (strictly positive) probabilities,
we seek a source $B$ for which $C$ achieves minimum redundancy, under the constraint that $X$ is a sub-source of $B$.

Let $X=\{x_1,\dots,x_m\}$ 
and let  $Y=\{y_1,\dots,y_{m-n}\}$ be the complementary sub-source of $X$ with respect to $B$.  For now, we assume that there is at least one unknown probability, i.e., $m<n$ and address the case when $m=n$ later. 

From~\eqref{a1}, the objective function $\mathcal{R}_{C}(B)$ is
\begin{equation}
\mathcal{R}_{C}(B)=\sum_{i=1}^n \p(b_i) \cdot \len_{C}(b_i) + \sum_{i=1}^n \p(b_i)\cdot \log \p(b_i).
\label{a14}
\end{equation}
Since the codeword lengths $\len_C(b_i)$ are known, 
and the probabilities $\p(b_i)$ are known for the symbols that lie in $X \sqsubseteq B$,
\eqref{a14}
can be rewritten as
\begin{equation}
\mathcal{R}_{C}(B)= \beta_0 +
\sum_{i=1}^{n-m} \p(y_i)\cdot \big(\beta_i + \log \p(y_i) \big),
\label{eq:rewritten}
\end{equation}
where $\beta_0$ and $\beta_i$ are constants, given by
\begin{equation}
\beta_0 = \sum_{i=1}^m \p(x_i) \cdot \big( \len_C(x_i) + \log \p(x_i) \big),
\end{equation}
and
\begin{equation}
\beta_i = \len_C(y_i) \text{\hspace{2.5em}for } 1 \le i \le n - m.
\end{equation}
Defining the additional constant
\begin{equation}
\beta_T = 1 - \sum_{i=1}^m {\p(x_i)},
\label{eq:betat}
\end{equation}
the minimization problem can be formally posed as
the following inequality constrained convex minimization problem:
\begin{equation}
\begin{aligned}
\text{minimize\hspace{2em}} & \beta_0 + \sum_{i=1}^{n-m} \p(y_i)\cdot \big(\beta_i + \log \p(y_i)\big) \\
\text{subject to\hspace{2em}} & 0 < \p(y_i) \le 1, \text{\hspace{2em}} 1 \le i \le n-m \\
\text{and\hspace{2em}} & \sum_{i=1}^{n-m} {\p(y_i)} = \beta_T.
\end{aligned}
\label{eq:convex_problem}
\end{equation}
It is worth noting that $\beta_T$ represents the total probability of the symbols with unknown probabilities. Under our current assumption that $m<n$,
$\beta_T$ is strictly positive.

To solve the constrained minimization problem, 
we initially ignore the inequality constraints and proceed via the method of Lagrangian multipliers (as in~\cite{MPK06}).

Setting
\begin{multline}
\frac{\partial}{\partial\, \p(y_i)}
\left( \beta_0 + \sum_{j=1}^{n-m} \p(y_j)\cdot \Big(\beta_j + \log \p(y_j)\Big) \right) 
\\
 - \lambda \cdot 
\frac{\partial}{\partial\, \p(y_i)} 
 \left(\sum_{j=1}^{n-m} {\p(y_j)} - \beta_T\right) = 0,
\end{multline}
and solving for the Lagrange multiplier, we get
$\lambda = \beta_i + \log \p(y_i)  + \log \mathrm{e}$,
which can be rearranged into 
\begin{equation}
\p(y_i) = 2^{\lambda-\beta_i} / \mathrm{e}.
\label{eq18_py}
\end{equation}
Substituting this last expression into the equality constraint we obtain
$\sum_{i=1}^{n-m} {2^{\lambda-\beta_i} / \mathrm{e}} = \beta_T$, or
\begin{equation}
2^\lambda = \frac{\beta_T \cdot \mathrm{e}}{\sum_{i=1}^{n-m} 2^{-\beta_i}}.
\label{eq19_py}
\end{equation}
Substituting~\eqref{eq19_py} into~\eqref{eq18_py} yields
\begin{equation}
\p(y_i) = 2^{-\beta_i}\cdot \frac{\beta_T}{\sum_{j=1}^{n-m} 2^{-\beta_j}}.
\label{eq22}
\end{equation}
Since $0 < \beta_i < \infty$ for $1 \le i \le n - m$ and $0 < \beta_T \le 1 $, we can see that $0 < \p(y_i) \le 1$. Hence, we have obtained a solution to the problem posed in~\eqref{eq:convex_problem}.

Substituting~\eqref{eq22} into~\eqref{eq:rewritten} yields the minimum redundancy: 
\begin{equation}
\mathcal{F}(X,C)= \beta_0 + \beta_T \cdot \log \frac{\beta_T}{\sum_{i=1}^{n-m} 2^{-\beta_i}}.
\label{eq:closed-form}
\end{equation}

It is easy to see that the algorithm discussed above is still correct for the case of $m=n$ that was left untackled above. In this case, \eqref{eq:rewritten} and~\eqref{eq:14} reveal that $\mathcal{F}(X,C)=\beta_0$. Noting that~\eqref{eq:betat} gives $\beta_T = 0$ and defining $0\log 0$ to be $0$ (as is often done in the literature), the convex optimization result of~\eqref{eq:closed-form} also yields $\mathcal{F}(X,C)= \beta_0$.

At this point, we have a method for computing a tight lower bound to the redundancy of a Huffman code for a source having $n \ge 2$ symbols, where the occurrence probability is known for only $m$ of the symbols, $0\le m \le n$. This method consists of applying the convex optimization described above to each code in $\Phi^{(n)}$, then taking the minimum over all such codes, as described by~\eqref{eq:rmin_calc}.

We would like to remark that, given $X$ and $C$, the result of \eqref{eq:closed-form} can be expressed as a
closed-form number, and thus the result of 
\eqref{eq:rmin_calc} is precisely calculable by a computer.

In subsequent sections, we discuss some practical issues, including the reduction of complexity that may arise when the cardinality of $\Phi^{(n)}$ is large. Before proceeding to these issues however, we generalize the method to the case when the size of the source is unknown. 

\section{General Bound} \label{sect3}

In this section, we formalize a method to obtain a bound on the minimum redundancy achievable by a Huffman code for a source of arbitrary size. That is, the bound holds for all $n \ge 2$ such that $n\ge m$. As before, $m$ is the number of symbols with known occurrence probabilities.

We begin with a simpler bound, $\mathcal{R}^{(\le n)}_\text{min}(X)$, 
on the lowest possible redundancy obtainable by a Huffman code
for any source $A$ having from $2$ to $n$ symbols, where $X \sqsubseteq A$ is the sub-source of $A$ containing the $m$ symbols with known probabilities.

\begin{definition} [Redundancy Bound for Sources of up to $n$ Symbols] 
Let $X$ be a sub-source of size $m$ for a source $A$, where $2 \le |A| \le n$.
Then
\begin{equation}
\mathcal{R}^{(\le n)}_\text{min}(X) = 
\min_{ \big\{ B\, \bigm\vert X \sqsubseteq B,\, 2 \le |B| \le n \big\} } 
\big\{\mathcal{R}_H(B) \big\}.
\end{equation}
\end{definition}

It is straightforward to see that $\mathcal{R}^{(\le n)}_\text{min}(X)$ 
can be obtained as 
\begin{equation}
\mathcal{R}^{(\le n)}_\text{min}(X) = 
\min_{ \big\{ k \, \bigm\vert m \le k,\,\, 2\le k \le n \big\}} \big\{
\mathcal{R}^{(k)}_\text{min}(X)
\big\},
\end{equation}
and that it is in fact tight. This simple bound paves the ground to obtain the desired bound, $\mathcal{R}^{*}_\text{min}(X)$, which is a tight bound 
on the lowest possible redundancy obtainable by a Huffman code
regardless of the cardinality of $A$.

\begin{definition}[General Redundancy Bound for Any Source] 
Let $X$ be a sub-source for a source $A$. We define
\begin{equation}
\mathcal{R}^{*}_\text{min}(X) = 
\min_{ \big\{ B\, \bigm\vert X \sqsubseteq B \big\} } 
\big\{\mathcal{R}_H(B) \big\}.
\label{eq:grb}
\end{equation}
\end{definition}

The calculation of $\mathcal{R}^{*}_\text{min}(X)$ is, in fact, the main objective of this manuscript. In the following theorem, we show that $\mathcal{R}^{*}_\text{min}(X)$ is equal to $\mathcal{R}^{(\le n)}_\text{min}(X)$ for all values of $n$ greater than or equal to a certain threshold $\mathcal{T}(X)$.

\newcommand{\rminstx}{\mathcal{R}^{\left(\le \mathcal{T}(X)\rule{0pt}{1.35ex} \right)}_\text{min}(X)}

\begin{theorem} 
$\mathcal{R}^{*}_\text{min}(X)$ can be obtained
as 
\begin{equation}
\mathcal{R}^{*}_\text{min}(X) = \rminstx,
\end{equation}
where
\begin{equation}
\mathcal{T}(X)=|X| + \left\lceil \frac{1 - \sum_{i=1}^{m} \p(x_i)}{
\min\big\{\p(x_1),\dots,\p(x_m)\big\} \rule{0pt}{2.3ex} 
} \right\rceil,
\label{a27}
\end{equation}
with the particular case of $\mathcal{R}^{*}_\text{min}(X)=0$ when $|X| = 0$.
\label{t:rmin_star}
\end{theorem}

Before proceeding to the proof of Theorem~\ref{t:rmin_star}, some
necessary lemmas are presented.

\begin{lemma}
Given a sub-source  $X$ and its complement $Y$, with respect to source $A$, with $n = |A| > \mathcal{T}(X)$,  $m = |X|\ge 1$,  and $n - m = |Y| \ge 1$, 
there are at least two elements in $Y$ with probabilities 
strictly smaller than $\min\{\p(x_1),\dots,\p(x_m)\}$.
\label{lemma:two_under_k}
\end{lemma}

\begin{proof}
By assumption, $|X|, |Y|\ge 1$ and $|A| \ge \mathcal{T}(X) + 1$. So,
\begin{multline}
|Y| = |A| - |X| \ge \mathcal{T}(X) + 1 - |X| = \\
= \left\lceil \frac{1 - \sum_{i=1}^{m} \p(x_i)}{
\min\big\{\p(x_1),\dots,\p(x_m)\big\} \rule{0pt}{2.3ex} 
} \right\rceil + 1.
\label{eq31}
\end{multline}
Since $|Y| \ge 1$, the numerator within the ceiling in~\eqref{eq31} is strictly positive
and thus $|Y|>1$, i.e., $Y$ has at least two elements.

Without loss of generality, we assume $\p(x_1)=\min\{\p(x_1),\dots,\p(x_m)\}$.
Suppose to the contrary that $\p(x_1) \le \p(y_i)$ for all $y_i \in Y$.
Combining $\sum_{i=1}^{m} \p(x_i) + \sum_{i=1}^{n-m} \p(y_i) = 1$ with the previous supposition, it follows that 
\begin{equation}
\sum_{i=1}^{m} \p(x_i) + \left|Y\right| \cdot \p(x_1) \le 1
\end{equation}
and so
\begin{equation}
|Y| \le \frac{1 - \sum_{i=1}^{m} \p(x_i)}{\p(x_1)}.
\label{a29}
\end{equation} 
 From~\eqref{eq31}, 
\begin{equation}
|Y| \ge \mathcal{T}(X) + 1 - |X|.
\label{e30}
\end{equation}
Substituting~\eqref{a29} and~\eqref{a27} into~\eqref{e30} yields
\begin{equation}
\frac{1 - \sum_{i=1}^{m} \p(x_i)}{ \p(x_1)} \ge \left\lceil \frac{1 - \sum_{i=1}^{m} \p(x_i)}{\p(x_1)} \right\rceil + 1,
\end{equation}
which is impossible. Hence, $\p(x_1) > \p(y_j)$ for at least some index $j$.

Suppose now that $\p(x_1) > \p(y_j)$ but that $\p(x_1) \le \p(y_i)$ for all $i \ne j$.
Again, combining $\sum_{i=1}^{m} \p(x_i) + \sum_{i=1}^{n-m} \p(y_i) = 1$ with the supposition, it follows that
\begin{equation}
\sum_{i=1}^{m} \p(x_i) + (\left|Y\right| - 1) \cdot \p(x_1) + \p(y_j) \le 1
\end{equation}
and
\begin{equation}
|Y| \le \big(1 - \sum_{i=1}^{m} \p(x_i) + \p(x_1) - \p(y_j)\big) / \p(x_1).
\label{a32}
\end{equation}
But substituting~\eqref{a32} and~\eqref{a27} into~\eqref{e30} yields
\begin{equation}
\big(1 - \sum \p(x_i) + \p(x_1) - \p(y_j)\big) / \p(x_1) \ge \left\lceil \frac{1 - \sum \p(x_i)}{\p(x_1)} \right\rceil + 1, 
\end{equation}
which implies
\begin{equation}
\frac{1 - \sum \p(x_i)}{\p(x_1)} + \frac{\p(x_1) - \p(y_j)}{\p(x_1)} \ge \left\lceil \frac{1 - \sum \p(x_i)}{\p(x_1)} \right\rceil + 1,
\end{equation}
and 
\begin{equation}
\frac{1 - \sum \p(x_i)}{\p(x_1)} - \left\lceil \frac{1 - \sum \p(x_i)}{\p(x_1)} \right\rceil \ge 1 - \frac{\p(x_1) - \p(y_j)}{\p(x_1)} = \frac{\p(y_j)}{\p(x_1)},
\end{equation}
which is impossible because
the left-hand side is always~$0$ or less, while the right-hand side is strictly positive. Hence, $\p(x_1) > \p(y_i)$ for at least two elements in $Y$.
\end{proof}

\begin{lemma}
Given a source $A$ of $n \ge 3$ symbols
for which symbols 
$a_j, a_k \in A$ are merged in step~3 of Algorithm~\ref{a1},
the source $B = \left( A
\setminus \left\{ a_j, a_k \right\} \right)
\cup \left\{ b_q \right\}$, with
$\p(b_q)=\p(a_j)+\p(a_k)$,
has $\mathcal{R}_H(B)\le \mathcal{R}_H(A)$.
\label{lema:symbols_merged}
\end{lemma}

\begin{proof}
Let $\Theta^{(i)}_A$ and $\Theta^{(i)}_B$ be the state  of Algorithm~\ref{a1} at iteration~$i$ when it is executed, respectively, for sources $A$ and $B$.
Since $\p(a_j) < \p(b_q)$ and $\p(a_k) < \p(b_q)$,
Algorithm~\ref{a1}
performs the same merging steps for both $A$ and $B$ before the 
algorithm merges $a_j$ and $a_k$ from $A$ into $[a_j, a_k]$
at some iteration $l$. I.e., 
\begin{equation}
 \Theta^{(i)}_A = (\Theta^{(i)}_B \setminus \{b_q\}) \cup \{a_j,a_k\} \hspace{2em}\text{for } i\le l.
\end{equation}
Once $a_j$ and $a_k$ have been merged by the algorithm for $A$, states $\Theta^{(l+1)}_A$ and $\Theta^{(l)}_B$ have an equal number of symbols with identical occurrence probabilities. Thus, both algorithms continue to produce identical outcomes except that ``$[a_j,a_k]$'' is substituted by ``$b_q$'' in the Huffman code for $B$.
This implies that 
\begin{equation}
\len_H(a_j) = \len_H(a_k) = \len_H(b_q) + 1,
\label{e1}
\end{equation} which can be employed to prove that
$\mathcal{R}_H(B) \le \mathcal{R}_H(A)$, or equivalently
\begin{equation}
\sum_{i=1}^{n-1}\p(b_i)\cdot \Big(
\len_H(b_i) + \log \p(b_i) \Big)
\le
\sum_{i=1}^{n}\p(a_i)\cdot \Big( \len_H(a_i)
+ \log \p(a_i) \Big).
\label{e2}
\end{equation}
In the following, a series of operations are carried out, each of which holds only if~\eqref{e2} holds.

Canceling equal terms on both sides in~\eqref{e2} and using~\eqref{e1} yields
\begin{multline}
\p(b_q)\cdot\big(\len_H(b_q)
+ \log \p(b_q) \big) \\
\le 
(\p(a_j)+\p(a_k))\cdot(\len_H(b_q)+1)  \\
+ \p(a_j)\cdot\log \p(a_j) + \p(a_k)\cdot\log \p(a_k) 
\end{multline}
Substituting $\p(a_j) + \p(a_k) = \p(b_q)$ yields
\begin{multline}
 \p(b_q)\cdot\log \p(b_q)
- \p(a_j)\cdot\log \p(a_j) \\
- \p(a_k)\cdot\log \p(a_k) \le
\p(b_q).
\end{multline}
Dividing by $\p(b_q)$ on both sides and rearranging results in
\begin{equation}
- \frac{\p(a_j)}{\p(b_q)}\cdot\log\frac{\p(a_j)}{\p(b_q)} - \frac{\p(a_k)}{\p(b_q)}\cdot\log\frac{\p(a_k)}{\p(b_q)} \le 1.
\label{e4}
\end{equation}
The left hand side of~\eqref{e4} is the entropy $\mathcal{H}(D)$ of a binary source 
$D = \{d_1, d_2\}$, with
$\p(d_1) = \p(a_j) / \p(b_q)$ and
$\p(d_2) = \p(a_k) / \p(b_q)$, which can be no greater than $1$.
\end{proof}

We now proceed to prove Theorem~\ref{t:rmin_star}.

\begin{proof}
Consider first the case of $|X|=0$, when all probabilities are unknown. It is trivial to see that $\mathcal{R}^{*}_\text{min}(X) = 0$, as redundancy must be at least $0$ by definition, and
$B=\{b_1,b_2\}$ with $\p(b_1)=\p(b_2)=0.5$ is an example where $\mathcal{R}_H(B)=0$. 

Suppose now that $|X|>0$. We will show that
for any source $B$ such that $X \sqsubseteq B$ with $|B| > \mathcal{T}(X)$,
there is a source $C$ such that $X \sqsubseteq C$, with $|C| =  \mathcal{T}(X)$ having $\mathcal{R}_H(C) \le \mathcal{R}_H(B)$.
Thus, it is unnecessary to consider any source with $|B| > \mathcal{T}(X)$ in~\eqref{eq:grb}.

Let $Y$ be the complementary sub-source of $X$ with respect to $B$, and let $\p(x_1)=\min\{\p(x_1),\dots,\p(x_m)\}$.
By assumption, $|B| > \mathcal{T}(X)$ and so from~\eqref{a27}, $|X|<|B|$.
So from Lemma~\ref{lemma:two_under_k}, there exist $i$ and $j$ such that $\p(y_i) < \p(x_1)$ and $\p(y_j) < \p(x_1)$.
That is, there exist two
symbols in $Y$ with probabilities smaller than any symbol in $X$. Thus,
the first two symbols merged in step~2 of the Huffman Algorithm, as 
applied to $B$, must come from $Y$.
Without loss of generality we choose $y_i$ and $y_j$ to be these two symbols.

Since $B$ contains $y_i$ and $y_j$, it follows that $|B| \ge |X| + 2 \ge 3$, and so from Lemma~\ref{lema:symbols_merged}, we can conclude that 
the source $C=B \setminus \{y_i,y_j\} \cup \{c\}$ with $\p(c) = \p(y_i) + \p(y_j)$ has $\mathcal{R}_H(C) \le \mathcal{R}_H(B)$, with $X \sqsubseteq C$ and $|C|=|B|-1$.

Simply put, the Huffman code for source $C$ is less redundant than the Huffman code for source $B$. Thus, it is unnecessary to consider source $B$ in the minimization of~\eqref{eq:grb}.

This procedure can be repeated until $|C| = \mathcal{T}(X)$.%
\end{proof}

\section{Efficient Enumeration of Prefix-free Codes} \label{sect4}

The previously presented methods to obtain bounds on Huffman code redundancy
rely on Algorithm~\ref{a2} for the exhaustive enumeration of all possible Huffman codes $\Phi^{(n)}$ 
for an alphabet of size $n$. This enumeration rapidly becomes untenable, since $\left| \Phi^{(n)} \right| =  n!\cdot (n-1)! / 2^{n-1}$, with factorial functions dominating the growth speed. This is even more problematic for the $\mathcal{R}^{*}_\text{min}$ bound, as all codes 
need to be enumerated for all alphabet sizes $n$ up to $\mathcal{T}(X)$. While 
$\mathcal{T}(X)$ may be small when the elements in $X$ have large probabilities, 
\eqref{a27} implies that $\mathcal{T}(X)-|X|$ is inversely proportional to the smallest probability in $X$, and thus $\mathcal{T}(X)$  can be large for small probabilities. For example,
a case as simple as $X=\{x_1\}$ with $\p(x_1)=0.01$ yields $\mathcal{T}(X)=100$ and $\left| \Phi^{(n)} \right| \simeq 10^{284}$. Clearly, $\Phi^{(n)}$ cannot be calculated by exhaustive means.

In this section we show how to reduce the complexity of Algorithm~\ref{a2} 
by enumerating codes in $\Phi^{(n)}$
in a manner that allows us to efficiently eliminate
codes that are provably unnecessary
for the purpose of obtaining $\mathcal{R}^{*}_\text{min}(X)$,
and thus making feasible the actual bound calculation.

This is accomplished by effectively managing the internal states
of the algorithm via the following two strategies.
\begin{enumerate}
\item[(i)]
We disregard
states resulting from the merging of two symbols in 
the sub-source complementary to $X$, as it can be seen through Lemma~\ref{lema:symbols_merged}
that there is a complementary sub-source to $X$ with one less symbol having lower or equal redundancy.
For this reason, we stop considering any state trajectory as soon as there is a merging of symbols which 
does not involve either a symbol resulting from a previous merge or one of the original symbols in $X$.

For example, 
given sub-source $X=\{a_1\}$ with $\p(a_1)=0.4$,
Algorithm~\ref{a2} is employed twice to calculate $\mathcal{R}^*_\text{min}(X)$,
producing $\Phi^{(2)}$ and $\Phi^{(3)}$.
In this case, code $[a_1,[a_2,a_3]]$ is in $\Phi^{(3)}$, but
through Lemma~\ref{lema:symbols_merged} it can be seen that code $[a_1,a'_2]$ in
$\Phi^{(2)}$, with $\p(a'_2) = \p(a_2) + \p(a_3)$, has lower or equal redundancy.
\item[(ii)] 
We keep track of constraints on probabilities that arise during Algorithm~\ref{a2},
which allows us to prune state trajectories that can not yield any viable code.

For example, 
given sub-source $X=\{a_1,a_2\}$ with $\p(a_1)=\p(a_2)=0.4$,
Algorithm~\ref{a2} produces (among others) $\Phi^{(3)}$ over alphabet $A=\{a_1,a_2,a_3\}$.
In this case, code $[[a_1,a_2],a_3]$ is in $\Phi^{(3)}$. For this code,
$a_1$ and $a_2$ are
merged first, which implies that $\p(a_1) \le \p(a_3)$ and $\p(a_2)\le \p(a_3)$.
These inequalities together with $0 \le \p(a_i)\le 1$,
$\sum \p(a_i)=1$, $\p(a_1)=0.4$, and $\p(a_2)=0.4$
yield an inconsistent system of equations, i.e.,  
$\p(a_3)$ must be $0.2$ which is incompatible with $\p(a_1) \le \p(a_3)$.
Thus, the code cannot be a
Huffman code for any source that has $X$ as sub-source,
and it can be ignored in an efficient enumeration of such codes.
\end{enumerate}

In the remainder of this section we describe a 
modified version of Algorithm~\ref{a2} that applies these two strategies
to dramatically reduce implementation complexity.

\subsection{Extended State}

First, we define two concepts that enable 
the application of the aforementioned strategies.
The first definition is a partition of each state employed in Algorithm~\ref{a2}
into \emph{known} symbols and \emph{unknown} symbols.

\begin{definition}[State partition] 
Let $\Theta$ be a state (alphabet) and let $X$ be a sub-source.
Based on $X$, we define a partition of $\Theta$ into
two complementary subsets:
\begin{itemize}
\item[--] a set of \emph{known} symbols $K$, containing all symbols in 
$\Theta$ that are either in $X$ or that are the result of one or more merging operations, with at least one of the symbols involved being a symbol in $X$; and
\item[--] a set of \emph{unknown} symbols $U = \Theta \setminus K$.
\end{itemize}
Let $\kappa_i$ denote an element in $K$ and $u_i$ denote an element in $U$.
\end{definition}

The second definition is that of an \emph{extended state}.

\begin{definition}[Extended state] 
A triplet
\begin{equation}
\dddot{\Theta} = (K, s, Z)
\end{equation}
is an extended state, where
$K$ is an alphabet, $s$ is an integer, and $Z$ is a set of linear inequalities. 
Let functions $\mathcal{K}(\cdot)$, $\mathcal{S}(\cdot)$, and $\mathcal{Z}(\cdot)$
denote each of the elements of the triplet.
That is $\mathcal{K}(\dddot{\Theta})=K$, $\mathcal{S}(\dddot{\Theta})=s$, and $\mathcal{Z}(\dddot{\Theta})=Z$.
\end{definition}

Extended states replace states (alphabets) in Algorithm~\ref{a2},
by providing an alternative representation of the alphabet, and augmenting it with constraints.
Specifically, an extended state $\dddot\Theta$ contains 
the set of known symbols $K$ of an equivalent non-extended state $\Theta$ 
together with an integer $s$ 
which is used to indicate the number of symbols that
have been drawn from $U$ at any point in the merging process.
In this formulation, symbols in $U$ can be thought of as being created when needed.
In addition, an extended state also contains a set, $Z$, containing
all inequalities that have arisen through previous merging of symbols, so that an extended state with an inconsistent $Z$ can be discarded according to strategy~(ii).

Following with the idea that symbols in $U$ are created when needed in a merging operation,
we can consider
three cases for merging symbols $\theta_i, \theta_j \in \Theta$ in step~2 of Algorithm~\ref{a2}:
\begin{itemize}
\item[(a)] $\theta_i, \theta_j \in K$,
\item[(b)] $\theta_i \in K$ and $\theta_j \in U$ (or vice-versa), and
\item[(c)] $\theta_i, \theta_j \in U$.
\end{itemize}
Per strategy (i), it is unnecessary to ever consider case~(c). 
One new symbol in $U$ is created each time case~(b) is applied, 
and thus, integer $s$ is equivalent to how many case-(b) merges have been performed.
It is clear that even though the probabilities of the created symbols $u_i$ are unknown,
they must be non-decreasing ($\p(u_i) \le \p(u_{i+1})$) starting with the first symbol drawn from $U$,
which we denote as~$u_0$.

\subsection{State Transition Functions for Extended States}

The state transition function $\mathrm{h}(\cdot)$ is now modified to operate
over extended states
resulting in three different functions, each covering one of the three cases of 
symbol merging operations described in the previous subsection.

For case (a), the state transition function $\mathrm{h}_a\left(\dddot\Theta,i,j\right)$ 
merges elements $\kappa_i$ and $\kappa_j$ in $\mathcal{K}\left(\dddot\Theta\right)$
as follows
\begin{multline}
\mathrm{h}_a\left(\dddot\Theta,i,j\right) = \bigg(
\Big(\mathcal{K}\left(\dddot\Theta\right) \setminus \big\{ \kappa_i, \kappa_j \big\} \Big) \cup \big\{[\kappa_i, \kappa_j]\big\}%
,\,\,\, \\[0pt plus 3pt minus 0pt]
 \mathcal{S}(\dddot\Theta),\,\,\, 
\mathcal{Z}(\dddot\Theta) \cup  \Omega_1 \cup  \Omega_2 \cup  \Omega_3 \cup  \Omega_4
\bigg),
\end{multline}
with
\begin{equation}
\Omega_1 = \Big\{ 
\p(\kappa_i) \le \p(\kappa_l) \Bigm\vert \kappa_l \in \mathcal{K}(\dddot\Theta), l \ne j
\Big\},
\end{equation}
\vspace{0pt plus 3pt minus 0pt}
\begin{equation}
\Omega_2 = \Big\{
\p(\kappa_j) \le \p(\kappa_l) \Bigm\vert \kappa_l \in \mathcal{K}(\dddot\Theta), l \ne i
\Big\},
\end{equation}
\vspace{0pt plus 3pt minus 0pt}
\begin{equation}
\Omega_3 = \Big\{
\p(\kappa_i) \le \p\left(u_{\mathcal{S}(\dddot\Theta)}\right) \Big\},
\end{equation}
and
\vspace{0pt plus 3pt minus 0pt}
\begin{equation}
\Omega_4 = \Big\{
\p(\kappa_j) \le  \p\left(u_{\mathcal{S}(\dddot\Theta)}\right)
\Big\}.
\end{equation}

Sets $\Omega_1$ to $\Omega_4$ contain inequalities 
stating that 
no other symbol has smaller occurrence probability than $\p(\kappa_i)$ and $\p(\kappa_j)$. Elements in $\mathcal{K}(\dddot\Theta)$ are covered in $\Omega_1$ and $\Omega_2$,
while elements (implicitly) in $U$ are covered in $\Omega_3$ and $\Omega_4$.
As elements in $U$ are created in non-decreasing order of probability,
the inequality involving $u_{\mathcal{S}(\dddot\Theta)}$ covers any element in $U$ that may be created in the future.

For case (b), the state transition function  $\mathrm{h}_b\left(\dddot\Theta,i\right)$
merges element $\kappa_i$ in $\mathcal{K}\left(\dddot\Theta\right)$
with a newly created element~$u_{\mathcal{S}(\dddot\Theta)}$ in $U$
as
\begin{multline}
\mathrm{h}_b\left(\dddot\Theta,i\right) = \bigg(
\Big(
\mathcal{K}\left(\dddot\Theta\right) \setminus \big\{ \kappa_i \big\} \Big) \cup \big\{[\kappa_i, u_{\mathcal{S}(\dddot\Theta)}]\big\},\,\,\, \\[0pt plus 3pt minus 0pt]
\mathcal{S}(\dddot\Theta) + 1,\,\,\, 
\mathcal{Z}(\dddot\Theta) \cup  \Upsilon_1 \cup  \Upsilon_2 \cup  \Upsilon_3 \cup  \Upsilon_4
\bigg)
\end{multline}
with
\begin{equation}
\Upsilon_1 = \Big\{ 
\p(\kappa_i) \le \p(\kappa_l) \Bigm\vert \kappa_l \in \mathcal{K}(\dddot\Theta) \Big\},
\end{equation}
\vspace{0pt plus 3pt minus 0pt}
\begin{equation}
\Upsilon_2 = \Big\{
\p\left(u_{\mathcal{S}(\dddot\Theta)}\right) \le \p(\kappa_l) \Bigm\vert \kappa_l \in \mathcal{K}(\dddot\Theta), l \ne i
\Big\},
\end{equation}
\vspace{0pt plus 3pt minus 0pt}
\begin{equation}
\Upsilon_3 = \Big\{
\p(\kappa_i) \le \p\left(u_{\mathcal{S}(\dddot\Theta)+1}\right) \Big\},
\end{equation}
and
\vspace{0pt plus 3pt minus 0pt}
\begin{equation}
\Upsilon_4 = \Big\{
\p\left(u_{\mathcal{S}(\dddot\Theta)}\right) \le  \p\left(u_{\mathcal{S}(\dddot\Theta)+1}\right)
\Big\}.
\label{eq:psi4}
\end{equation}

It is unnecessary to give the state transition function for case (c), as per strategy (i), it only yields extended states not worth considering.

\subsection{Algorithm}

Following the aforementioned strategies and definitions, 
Algorithm~\ref{a2} is modified as follows.

Given a sub-source $X=\{x_1,\dots,x_m\}$,
the algorithm starts at the initial extended state of
$\big(\{x_1,\dots,x_m\},\,\, 0,\,\, \{0 \le \p(u_0) \} \big)$
and considers all possible state trajectories,
through repeated application of the state transition functions $\mathrm{h}_a(\cdot,\cdot,\cdot)$ and
$\mathrm{h}_b(\cdot,\cdot)$.

Once a given extended state $\dddot\Theta$ is reached with $|\mathcal{K}(\dddot\Theta)| = 1$, $\mathcal{K}(\dddot\Theta)$ contains a single symbol, which is equivalent to a prefix-free code of $|X| + \mathcal{S}(\dddot\Theta)$ codewords.
The redundancy for this code is lower bounded as described in Subsection~\ref{sec:minimization}.
State transition function $\mathrm{h}_a$ cannot be further applied to $\dddot\Theta$,
but state transition function $\mathrm{h}_b$ may still be applied, producing 
a new state $\dddot\Theta'$
representing a prefix-free code of $|X| + \mathcal{S}(\dddot\Theta) + 1$ codewords.
Thus, finding an extended state equivalent to a prefix-free code, does not terminate a state trajectory.

State trajectories are terminated when $|X| + \mathcal{S}(\dddot\Theta) > \mathcal{T}(X)$ (as per Theorem~\ref{t:rmin_star}); thus
the algorithm is guaranteed to stop. In addition, as per strategy~(ii),
state trajectories are terminated when $\mathcal{Z}(\dddot\Theta)$ becomes inconsistent, significantly limiting the number of extended states to examine.

These modifications yield a new algorithm capable of 
generating a set of codes sufficient to calculate
 $\mathcal{R}^*_\text{min}(\cdot)$ without having to fully enumerate $\Phi^{(n)}$.

\begin{algo} 
Let $X=\{x_1,\dots,x_m\}$ be a given sub-source of $m$ symbols, let $\Psi$ be a set of extended states and let $i$ be an iteration index.
\begin{enumerate}
\item Let $\dddot\Phi^{(0)} = \Big\{ \big(\{x_1,\dots,x_m\},\,\, 0,\,\, \{0 \le \p(u_0) \} \big) \Big\}$, set $\Psi \gets \emptyset$, and set $i \gets 0$.
\item Let $\Pi_1= \Big\{\mathrm{h}_a(\dddot\Theta, j, k) 
\Bigm\vert \dddot\Theta \in \dddot\Phi^{(i)},\,\,
 1 \le j < k \le |\mathcal{K}(\dddot\Theta)|
  \Big\}$.
\item Let $\Pi_2 = 
\Big\{\mathrm{h}_b(\dddot\Theta, j) 
\Bigm\vert \dddot\Theta \in \dddot\Phi^{(i)},\,\,
1 \le j \le |\mathcal{K}(\dddot\Theta)|,\,\,
|X| + \mathcal{S}(\dddot\Theta) < \mathcal{T}(X)
  \Big\}  
  $.
\item $\dddot\Phi^{(i+1)} = \Big\{\dddot\Theta \Bigm\vert \dddot\Theta \in \Pi_1 \cup \Pi_2$ such that 
$\mathcal{Z}(\dddot\Theta)$ is not inconsistent$\Big\}$
\item $\Psi \gets \Psi \cup \Big\{ \dddot\Theta \Bigm\vert  
\dddot\Theta \in \dddot\Phi^{(i+1)} ,\,\,
|\mathcal{K}(\dddot\Theta)| = 1 \Big\}$
\item Set $i \gets i + 1$.
\item If $\dddot\Phi^{(i)} \ne \emptyset$, go to step 2.
\item Stop.
\end{enumerate}
\label{a3}
\end{algo}

Once Algorithm~\ref{a3} stops, $\Psi$ is a set of extended
states containing 
sufficient codes to yield the least possible redundancy
for every source $B$ having sub-source $X$.
$\mathcal{R}^*_\text{min}$
is then computed by applying the convex optimization of Subsection~\ref{sec:minimization}
to each code in $\Psi$, and taking
the minimum of all such minimization results. This is facilitated by the fact that $\Psi$
is only a small subset of 
\begin{equation}
\bigcup_{i=2}^{\mathcal{T}(X)} \Phi^{(i)}, \label{eq:cupphi}
\end{equation}
which is what would need to be examined in the exhaustive case.

The reduction in the number of elements in $\Psi$ as compared to those in~\eqref{eq:cupphi} is substantial,
even if not easily given in closed form. Following the example at the beginning of this section, for the case of $X=\{x_1\}$ with $\p(x_1)=0.01$, we obtain $|\Psi| = 11$, which is clearly smaller than $\left| \Phi^{(100)} \right| \simeq 10^{284}$ and, by extension, $\left|\bigcup_{i=2}^{100} \Phi^{(i)}\right|$. Examining 
$11$ cases can be carried out in negligible time, while examining $10^{284}$ cases 
at a reasonable rate may take considerably more than the age of the universe.
We show further evidence of the efficiency of the proposed method in Section~V.
Note, however, that we do not aim to provide a theoretical complexity result for 
the proposed algorithm, as the details of such a study would require details of 
the structure of the Huffman code redundancy which we are only
starting to uncover in this paper (e.g., how many extended states have inconsistent constraints
in step 4 of Algorithm~3).
Moreover, when accounting for algorithmic complexity reductions,
not only search space reductions need to be accounted for, but the effort to prune the search space. Otherwise, the search space could be reduced at the expense of a more complex pruning stage.

We conclude this Subsection with one more remark in relation to Algorithm~3.
Given the equivalence for a binary prefix code of it being a Huffman code and of 
that code having the sibling property~\cite{Gal78}, 
we could have considered the constraints imposed by the sibling property to obtain the
necessary conditions in $\mathcal{Z}(\dddot\Theta)$ for $\dddot\Theta$ to be a Huffman code, even if in a less straightforward manner.

\subsection{Consistency Verification}

The final consideration in this section is that of evaluating
whether a system of inequality constraints is consistent or not. We have intentionally overlooked the issue up to this point, as it is a self-contained problem. In this subsection, we formalize the problem, and then describe how to solve it efficiently.

Given an extended state $\dddot\Theta$ we want to know whether the following system of
inequalities is consistent:
\begin{multline}
\big\{ \p(x_i) = P_i \big\}_{\forall x_i \in X}
\cup\,\,
\mathcal{Z}(\dddot\Theta)
\,\,\cup
\Big\{ 
\sum_{i=1}^{|X|} \p(x_i) \\
+ \sum_{i=0}^{\mathcal{S}(\dddot\Theta)-1} \p(u_i) \le 1,\,\, 
 \p(u_{\mathcal{S}(\dddot\Theta)}) = 1 \Big\}
\label{eq:ineq}
\end{multline}
The system is composed by all known probabilities given by $X$, here denoted by $P_i$,  by all constraints due to merging operations in $\mathcal{Z}(\dddot\Theta)$,
and by two more additional constraints. The first additional constraint,  
$\sum \p(x_i) + \sum \p(u_i) \le 1$, ensures that probabilities for a prefix-free code deriving from $\dddot\Theta$ do not grow over $1$. It is not a strict equality, as additional symbols from $U$ may be added to $\dddot\Theta$ by $\text{h}_b$. The second additional constraint, $\p(u_{\mathcal{S}(\dddot\Theta)}) = 1$, together with $0 \le \p(u_0)$ and 
$\p(u_i) \le \p(u_{i+1})$ from $\mathcal{Z}(\dddot\Theta)$, ensure that $0 \le \p(u_0) \le \p(u_1) \le \ldots \le \p(u_{\mathcal{S}(\dddot\Theta)-1}) \le 1$.

It is well known that the problem of determining the feasibility of the previous inequality system can be posed as a linear programming problem as follows.
\begin{equation}
\begin{aligned}
\text{minimize\hspace{2em}} & \text{f}(\mathbf{x}) \\
\text{subject to\hspace{2em}} & 
\mathbf{A} \mathbf{x} - \mathbf{b} \preceq \mathbf{0} \\
\end{aligned}
\label{eq:linear_program}
\end{equation}
where $\mathbf{x} = \big(\p(u_0), \ldots, \p(u_{\mathcal{S}(\dddot\Theta)-1})\big)$, $\text{f}(\cdot)$ is an arbitrary linear function of $\mathbf{x}$, and $\mathbf{A} \mathbf{x} - \mathbf{b} \preceq \mathbf{0}$ are inequalities equivalent to those in~\eqref{eq:ineq}, with $\preceq$ denoting element-wise comparison.
We employ bold symbols to denote vector and matrices, and to distinguish them from those of previous sections.

Finding any solution, $\mathbf{x}^*$, for~\eqref{eq:linear_program} yields a point satisfying all inequalities. Conversely, if there is no solution to~\eqref{eq:linear_program}, then the system in~\eqref{eq:ineq} is inconsistent.
In our case, solving~\eqref{eq:linear_program} is not straightforward.
Up to this point, a CAS
can perform all operations described in this manuscript,
thus ensuring that solutions can be found 
for all redundancy bounds.
Unfortunately, efficient linear program solvers are of a numerical nature,
and as such the optimization can incorrectly fail to reach a solution due to numerical problems,
which could incorrectly discard valid extended states in Algorithm~\ref{a3},
thus resulting in an incorrect calculation of $\mathcal{R}^*_\text{min}$.
Note that the opposite could also occur, but does not represent a serious problem. 
That is, an invalid solution could be reached, which could incorrectly preserve an inconsistent extended state, but this would only increase the size of $\Psi$ and not alter 
the validity of its use to find $\mathcal{R}^*_\text{min}$.

\newcommand{\blambda}{\text{\boldmath$\lambda$}}

To address the issue of numerical solvers, we can reformulate the linear program, through duality~\cite{BV04}, into
a problem in which finding any solution (regardless of optimality) proves 
that no solution is possible for~\eqref{eq:linear_program}.
First, we reformulate~\eqref{eq:linear_program} into the following so-called Phase I problem:
\newcommand{\s}{\gamma}
\begin{equation}
\begin{aligned}
\text{minimize\hspace{2em}} & \s \\
\text{subject to\hspace{2em}} & 
\mathbf{A} \mathbf{x} - \mathbf{b} \preceq \s \cdot \mathbf{1} \\
\\
\end{aligned}
\label{eq:linear_program_p1}
\end{equation}
If there exists a solution $(\s^*, \mathbf{x}^*)$ with $\s^* > 0$, then there is no $\s$ less than $\s^*$ that satisfies the constraints and thus $\mathbf{A} \mathbf{x} - \mathbf{b} \preceq \mathbf{0}$ cannot be valid.

For~\eqref{eq:linear_program_p1}, a numerical solver would yield an approximate solution $(\widetilde{\s}^*, \widetilde{\mathbf{x}}^*)$
from which it would not be possible to infer whether $\s^* > 0$ is true or not.
To address this, we employ the dual problem of~\eqref{eq:linear_program_p1}:
\begin{equation}
\begin{aligned}
\text{maximize\hspace{2em}} & \text{g}(\blambda) \\
\text{subject to\hspace{2em}} & \mathbf{A}^{T} \blambda = 0 \\
& \mathbf{1}^{T} \blambda = 1 \\
& \blambda \succeq \mathbf{0}, \\
\end{aligned}
\label{eq:linear_program_pd}
\end{equation}
with $\text{g}(\blambda) = - \mathbf{b}^{T} \blambda$ (see~\cite[p. 225]{BV04}).

From strong duality and Slater's constraint qualification, 
we can conclude that $\text{g}(\blambda) \le \s$
and that $\text{g}(\blambda^*) = \s^*$ for solution $\blambda^*$.
If solving~\eqref{eq:linear_program_pd} yields $\text{g}(\blambda^*) > 0$, 
then $\s^* > 0$ and we have proved that~\eqref{eq:ineq} is inconsistent. Or for that matter, 
\eqref{eq:ineq} can be declared inconsistent as soon as any $\blambda$ is found that verifies $\text{g}(\blambda) > 0$ during the optimization process, since
we have found proof that $\s \ge \text{g}(\blambda) > 0 $.

Again, a numerical solver yields only an approximate solution $\widetilde{\blambda}^*$ for~\eqref{eq:linear_program_pd}. However, this is sufficient for our purposes,
as the consequences of finding an incorrect solution and of failing to reach a solution are
interchanged from primal to dual problems.
The solution given by the numerical solver is a vector of rational numbers
which can be safely verified with a CAS. Hence, 
if a solution yields $\text{g}(\widetilde{\blambda}^*) > 0$ 
we can conclude that~\eqref{eq:ineq} is inconsistent and terminate
the corresponding state trajectory.
Should a solution fail to satisfy $\text{g}(\widetilde{\blambda}^*) > 0$ 
or the numerical method fail to 
reach a solution, we cannot make claims regarding the consistency of~\eqref{eq:ineq}.
In this case we cannot terminate the state trajectory.

\section{Examples} \label{sect5} %

Examples and applications of the results obtained by the general bound $\mathcal{R}_\text{min}^*$ are presented in this section.

\subsection{General Examples}

\begin{figure}[h]
\begin{center}
\includegraphics[width=\linewidth]{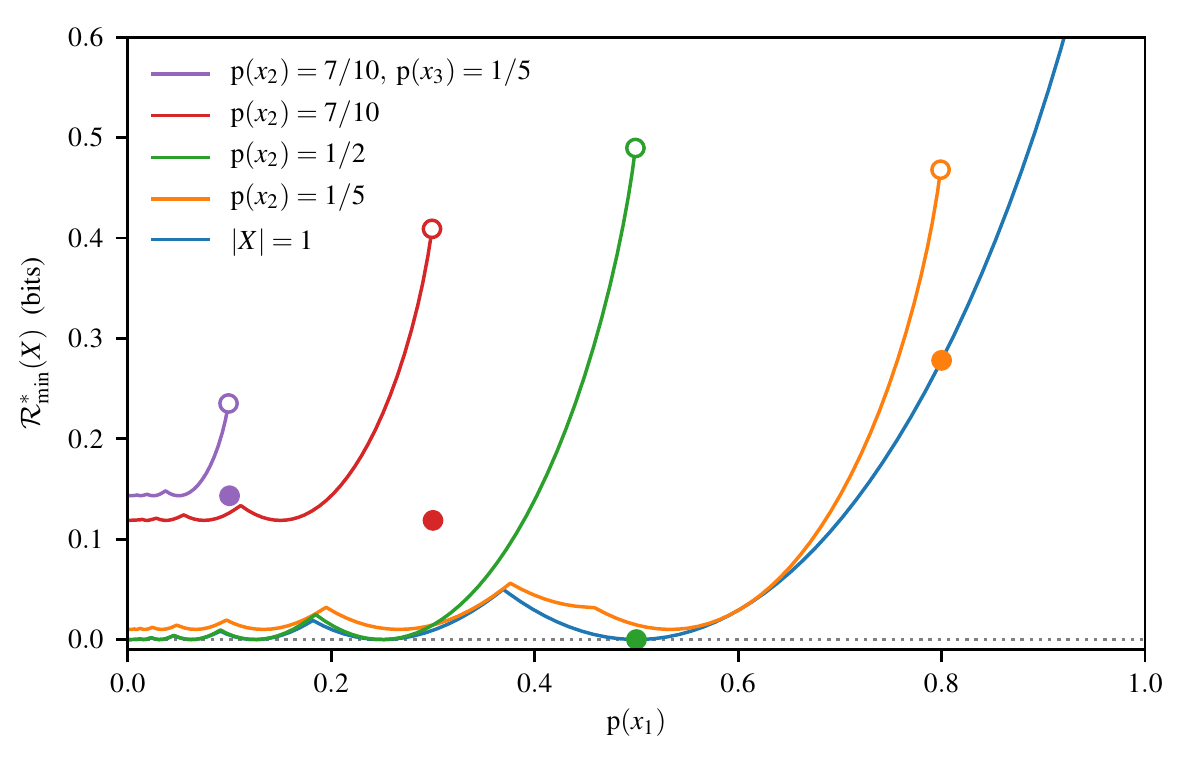}
\end{center}
\caption{Plot of $\mathcal{R}^*_{\mathrm{min}}(X)$, where some
probabilites are known, with $\p(x_1)$ varying along the horizontal axis.}
\label{f1}
\end{figure}

The bound $\mathcal{R}_\text{min}^*$ is plotted in Fig.~\ref{f1} for a few cases 
of one, two, or three known probabilities. 
For the case of one known probability, 
the best bounds available in the literature are given in 
\cite[Fig.~1]{MA87} and \cite[Fig.~4]{MPK06}.
As expected, the proposed method yields results that are consistent with the results reported in those
works. In particular, a computer algebra system was employed to verify that the bounds evaluate
identically on a fine grid of probability values.

It can be seen that as the sum of the known probabilities approaches~$1$,
redundancy quickly rises. Redundancy falls again, once the sum reaches~$1$.
The rise is caused by the prefix-free code having to necessarily account for one additional
symbol with very low associated probability.

For the cases of $\p(x_2)=7/10$ and $\p(x_2)=1/2$,
the respective curves are identical
to that of $|X|=1$, except for constant shift along the y-axis and a difference in scale.
It can easily be seen that this is the
result of $x_2$ being merged by the last step of the Huffman Algorithm in both cases.
In fact, this can be seen to be true whenever the known 
probabilities that can be arranged in a sequence 
where each successive probability 
is larger than the sum of the remaining ones (see \cite[Lemma~1]{MPK06}).
I.e., $\p(x_i) > 1 - \sum_{j \le i} \p(x_j)$.
For example, this holds for the case shown when $\p(x_2)=7/10$ and $\p(x_3)=1/5$, where $\frac{7}{10} > 1 - \frac{7}{10} = \frac{3}{10}$ and $\frac{1}{5} > 1 - \frac{7}{10} - \frac{1}{5} = \frac{1}{10}$.
However, this does not hold for the case when $\p(x_2)=1/5$, as 
evidenced by lack of smoothness at $\p(x_1)\simeq 0.4581$ caused by
the transition between optimal prefix-free codes.

We can use the examples in Fig.~\ref{f1} to disprove that 
\begin{equation}
\mathcal{R}_\text{min}^*(X) \simeq
\max_{x_i \in X}\Big\{\mathcal{R}_\text{min}^*(\{x_i\})\Big\}.
\end{equation}
This implies that considering known probabilities independently, one at a time, does not yield 
a good estimator of minimum redundancy.
For example, 
given $X=\{x_1,x_2\}$ with $\p(x_1)=0.49$ and $\p(x_1)=0.5$,
we have
\begin{equation}
\mathcal{R}_\text{min}^*\big(X\big) = \frac{49}{50}+\frac{49}{50}\log \frac{7}{10} - \frac{1}{50}\log 5 \simeq 0.4293 \text{ bits},
\end{equation}
while 
\begin{equation}
\mathcal{R}_\text{min}^*\big(\{x_1\}\big) = \frac{49}{100} +
\frac{49}{50}\log\frac{7}{10} + \frac{51}{100}\log \frac{51}{50}
\simeq 0.0003 \text{ bits}
\end{equation}
and
\begin{equation}
\mathcal{R}_\text{min}^*\big(\{x_2\}\big) = 0 \text{ bits}.
\end{equation}

In Fig.~\ref{f2}, four two-dimensional contour plots are presented, in which two known probabilities vary along the axes of the plot.
These plots are colored to denote regions sharing the same optimal prefix-free code.
As more than one prefix-free code may be optimal at some given coordinates,
codes covering larger regions take precedence over smaller ones in the figure.
The plot follows a ``map coloring'' scheme, where 
colors may repeat but adjacent codes are guaranteed to be of different color.
Consistent with the idea of Shannon coding, it can be seen that in Fig.~\ref{fig:fig1} local minimums are arranged at coordinates where $\p(x_1)$ and $\p(x_2)$ are negative powers of two (i.e.,
$\p(x_1)=2^{-i}$, $\p(x_2)=2^{-j}$ with $i,j\in \mathbb{N}$).
In this case, a single prefix-free code covers the local neighborhood
around a local minimum.
Contours for Fig.~\ref{fig:fig2} and~\subref{fig:fig4} are similar in character to those in~\subref{fig:fig1}, except for scale.
Fig.~\ref{fig:fig3} exhibits more substantial differences
in the regions covered by each prefix code,
with changes both in shape and quantity.
In Fig.~\ref{fig:fig1}, 
there is an inappreciable diagonal line
where~$\p(x_1) + \p(x_2) = 1$ and~$|Y| = 0$,
for which the optimal code is different from 
those in the interior of the plot where $|Y| > 0$.
This discontinuity is akin to those in Fig.~\ref{f1}. Similar 
discontinuities appear in parts~\subref{fig:fig2},~\subref{fig:fig3} and~\subref{fig:fig4}
of Fig.~\ref{f2}.

One additional remark regarding Fig.~\ref{f2} is that 
the optimal regions for prefix-free codes seem to be of polygonal shape,
and possibly convex.
The shape of these regions is an interesting topic for future research.

\begin{figure*}[h]
\begin{center}
\begin{subfigure}[b]{0.5\linewidth}%
\centering\includegraphics[width=.9\linewidth]{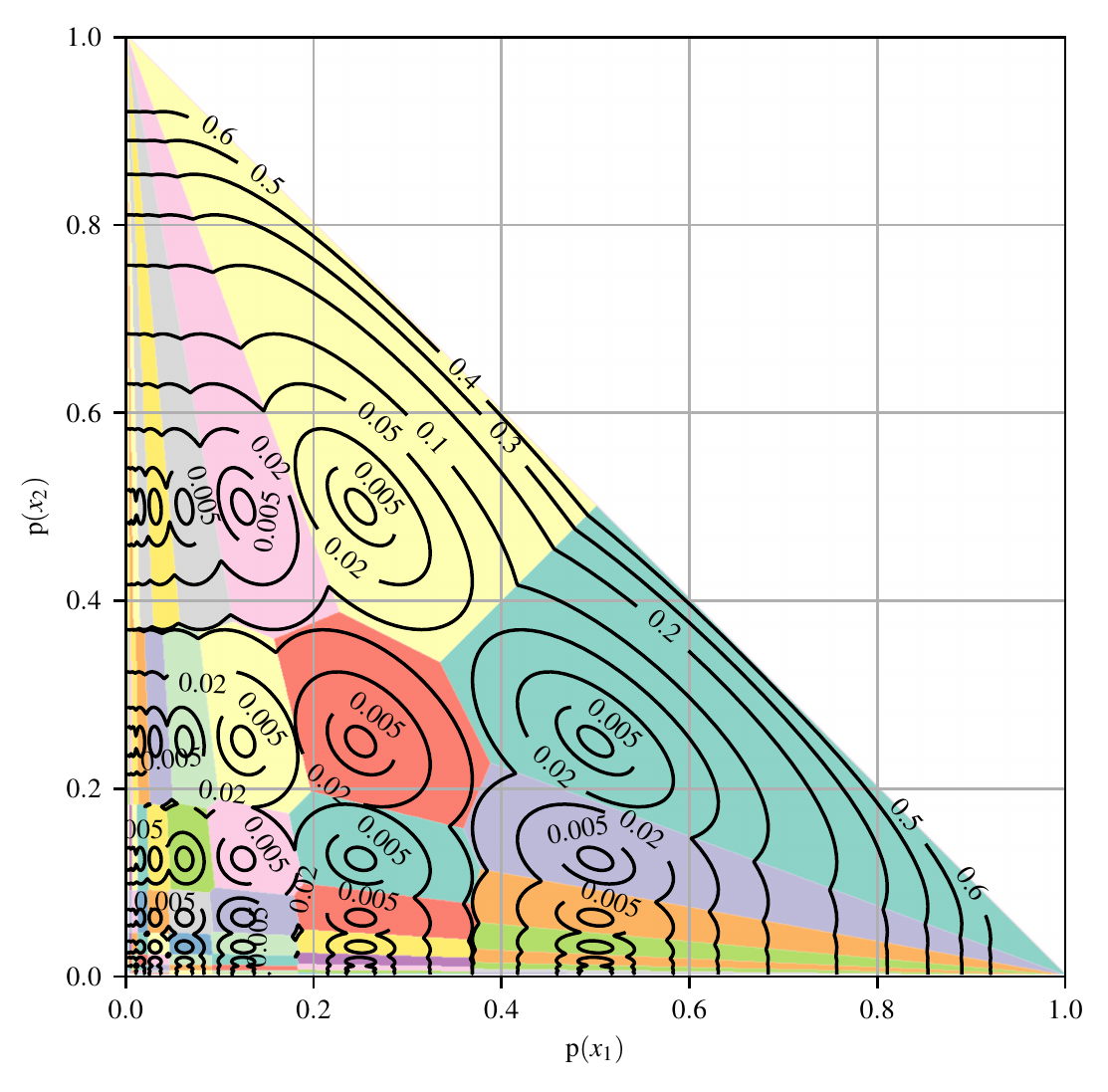}
\caption{$|X|= 2$ \label{fig:fig1}}%
\end{subfigure}%
\begin{subfigure}[b]{0.5\linewidth}%
\centering\includegraphics[width=.9\linewidth]{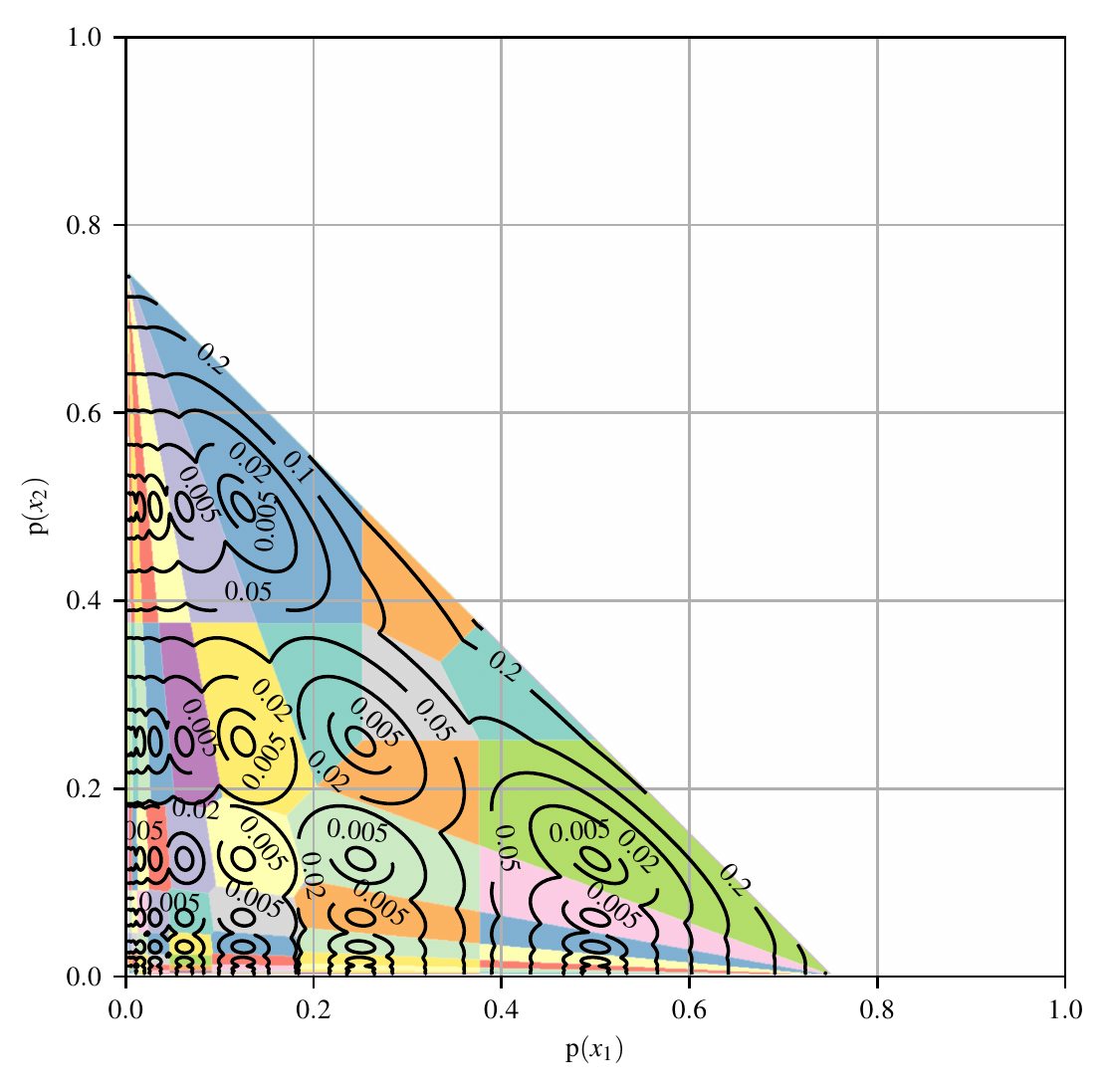}
\caption{$|X|= 3$, $\p(x_3) = 1/4$\label{fig:fig2}}%
\end{subfigure}\\[10pt]

\begin{subfigure}[b]{0.5\linewidth}%
\centering\includegraphics[width=.9\linewidth]{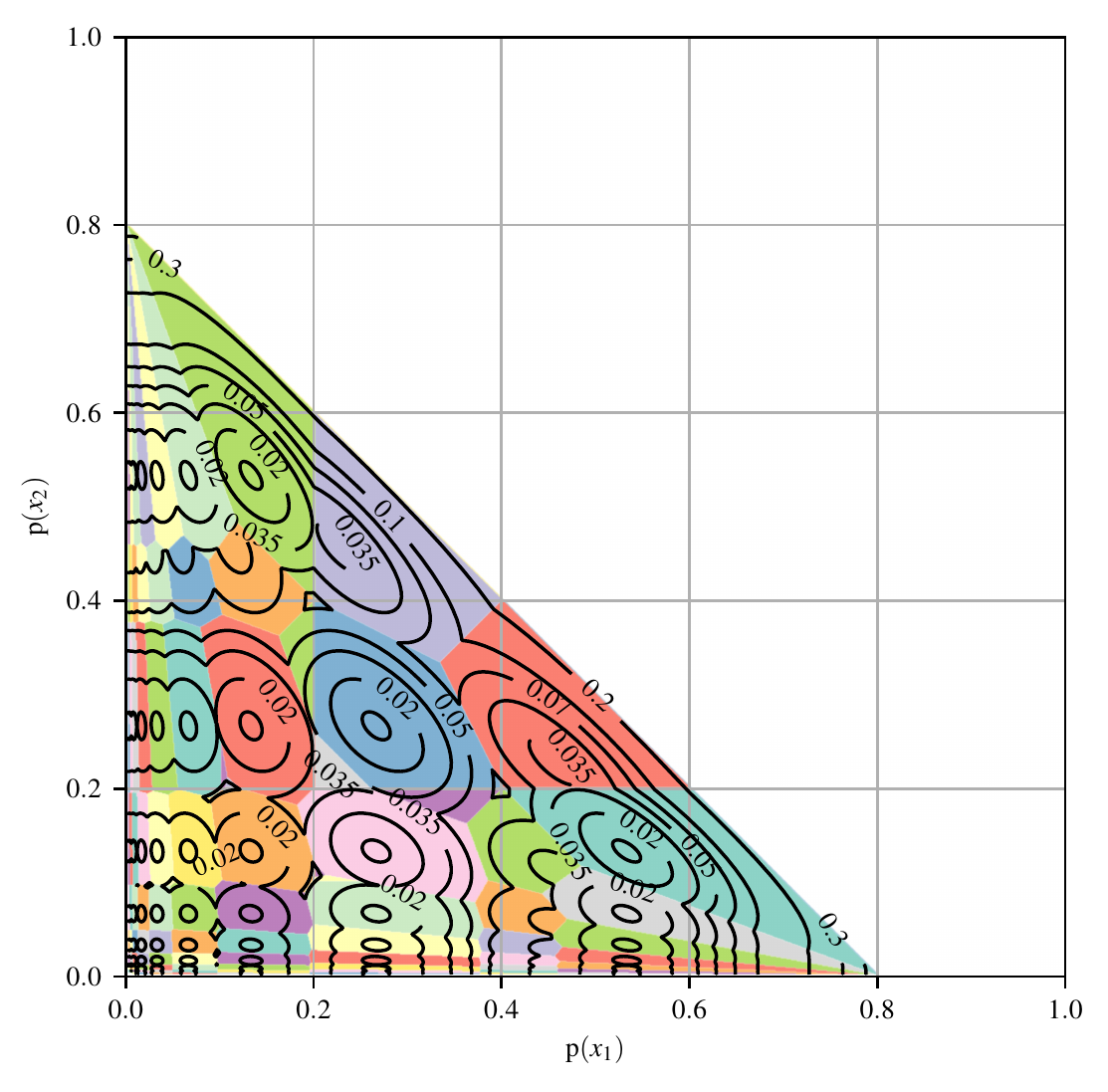}
\caption{$|X|= 3$, $\p(x_3) = 1/5$\label{fig:fig3}}%
\end{subfigure}%
\begin{subfigure}[b]{0.5\linewidth}%
\centering\includegraphics[width=.9\linewidth]{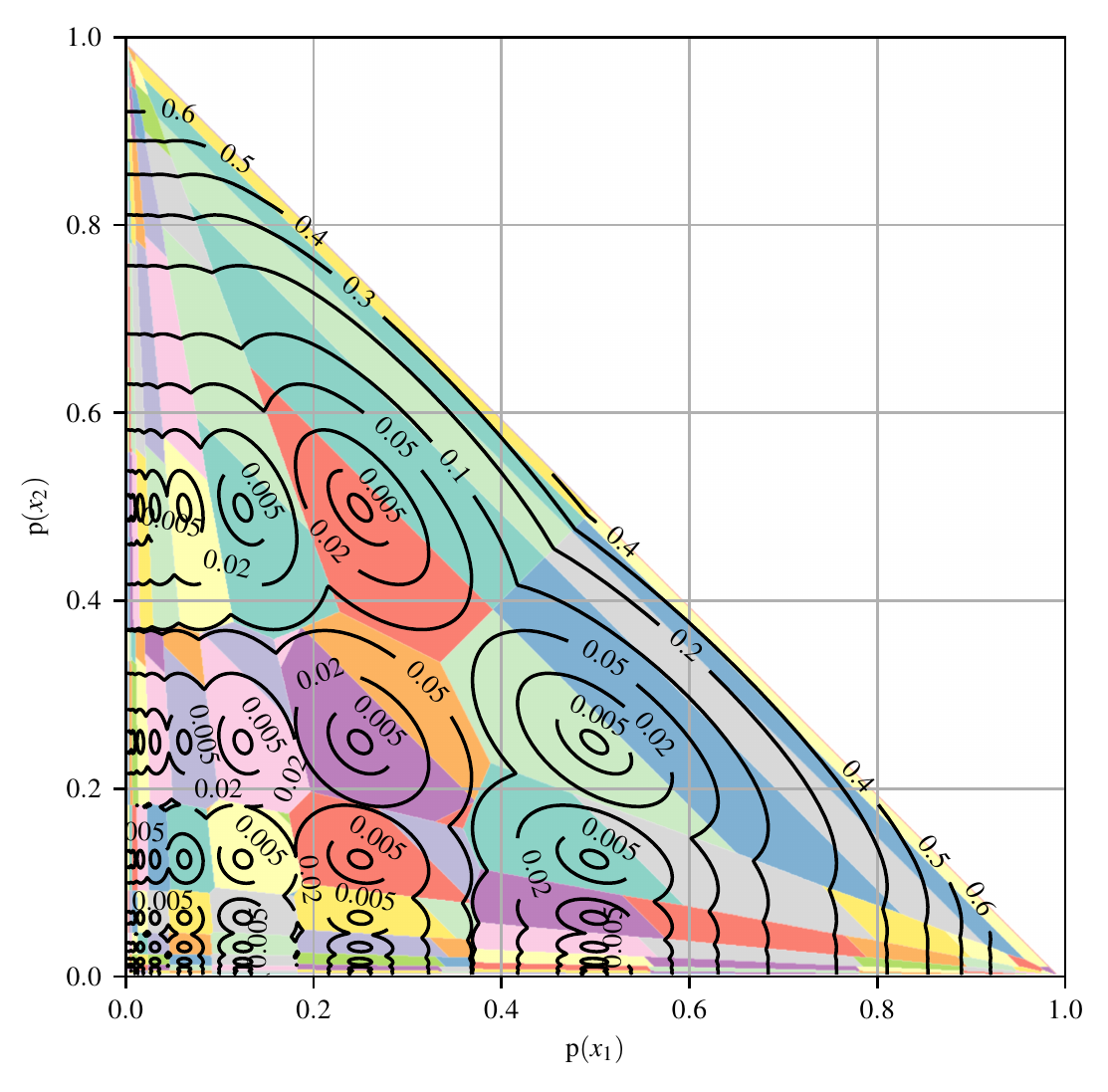}
\caption{$|X|= 3$, $\p(x_3) = 1/100$\label{fig:fig4}}%
\end{subfigure}%
\end{center}
\caption{Contour plots of $\mathcal{R}^{*}_{\mathrm{min}}(X)$, where $\p(x_1)$ varies along the horizontal axis and $\p(x_2)$ varies along the vertical axis.
Background colors denote different Huffman codes.
Adjacent codes have different colors, but colors may repeat for
non-adjacent codes. Labels on contour lines are in bits.}
\label{f2}
\end{figure*}

\subsection{Closed-form Expression for Two Known Probabilities}

By examining codes employed to produce Fig.~\ref{fig:fig1},
a common pattern can be found for the case of $|X|=2$.
For each local minimum at $\p(x_1)=2^{-a}$, $\p(x_2)=2^{-b}$,
at least one of the optimal codes 
associated with that minimum has the structure reported in Fig.~\ref{fig:drawing}.
Code optimality at the local minimum can easily be proven by setting $\p(y_i)$ to the probability of the sibling of $y_i$.
For $\p(x_1) < 0.5$ and $\p(x_2) < 0.5$, 
the aforementioned codes at the four (or three) local minimums surrounding a given pair of $\p(x_1)$ and $\p(x_2)$
seem sufficient to obtain the lower redundancy bound
for that point.

\begin{conjecture}\label{t3}
For $|X|=2$, $\mathcal{R}_\text{min}^*\big(\{x_1,x_2\}\big)$
is given by one of the three following expressions.
\begin{enumerate}
\item[(a)] When $\p(x_1) + \p(x_2) = 1$, Eq.~\eqref{a1} can be directly applied:
\begin{equation}
\mathcal{R}_\text{min}^*\big(\{x_1,x_2\}\big) = 
\mathcal{R}_{[x_1,x_2]}\big(\{x_1,x_2\}\big).
\end{equation}

\item[(b)] When $\p(x_1) + \p(x_2) < 1$ and $\p(x_1) \ge 0.5$ (and similarly when $\p(x_2) \ge 0.5$),
from Lemma~1 in~\cite{MPK06} we have that
\begin{equation}
\mathcal{R}_\text{min}^*\big(\{x_1,x_2\}\big) = 
\mathcal{R}_\text{min}^*\big(\{x_1\}\big) +
(1 - \p(x_1)) \cdot \mathcal{R}_\text{min}^*\big(\{x'\}\big)
\end{equation}
where $\p(x') = \frac{\p(x_2)}{1-\p(x_1)}$.
Note that the expression for $\mathcal{R}_\text{min}^*\big(\{x\}\big)$ is given in~\cite{MPK06}.

\item[(c)] Otherwise,
\begin{equation}
\mathcal{R}_\text{min}^*\big(\{x_1,x_2\}\big) = 
\min_{C \in \Delta} \mathcal{F}\big(\{x_1,x_2\}, C\big).
\end{equation}
where $\Delta$ contains the four possible codes
given by Fig.~\ref{fig:drawing} when $a \in \{\lfloor - \log \p(x_1) \rfloor, \lceil - \log \p(x_1) \rceil\}$ and $b \in \{ \lfloor - \log \p(x_2) \rfloor, \lceil - \log \p(x_2) \rceil \}$, with the exception of code $[x_1,x_2]$, which is never in $\Delta$.
That is
\begin{multline}
\mathcal{R}_\text{min}^*\big(\{x_1,x_2\}\big) = \beta_T \cdot \log \beta_T \\
+ \min_{C_{(a,b)} \in \Delta} 
\Big\{
\beta_0 - \beta_T \cdot \log \left({1 - 2^{-a} - 2^{-b}}\right) \Big\}.
\end{multline}
\end{enumerate}
\end{conjecture}

We have seen the conjecture hold for $0.001 \le \p(x_i) \le 0.999$ at $0.001$ intervals,
and for $0.05 \le \p(x_1) \le 0.9999$ and $0.05 \le \p(x_2) \le 0.1$ at $0.0001$ intervals.

Interestingly, the conjecture draws parallels with the expression for $m=1$ in~\cite{MPK06}, in that the minimal redundancy for $m=1$ is as if chosen from two neighboring codes with well-known structures and with minimum redundancies at $2^{\lfloor -\mathrm{log}\, \mathrm{p}(x_1) \rfloor}$ and $2^{\lceil -\mathrm{log}\, \mathrm{p}(x_1) \rceil}$.

\begin{figure}
\begin{center}
\includegraphics[width=.8\linewidth]{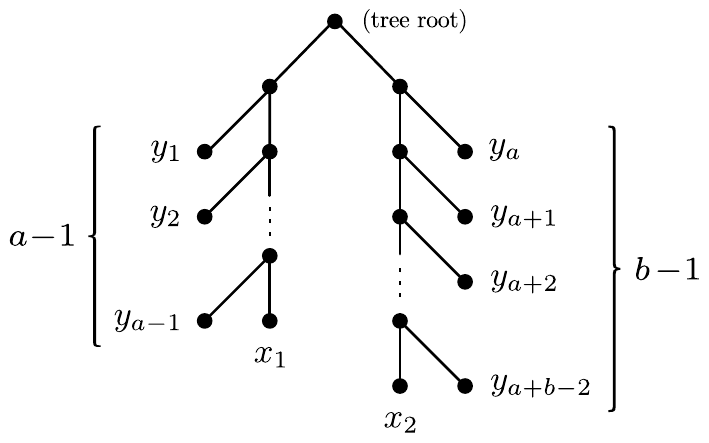}
\end{center}
\caption{Optimal code for $|X|=2$ associated with a local minimum at $\p(x_1)=2^{-a}$, $\p(x_2)=2^{-b}$.}\label{fig:drawing}
\end{figure}

\subsection{Run-time Analysis}

In this subsection we provide a brief experimental  run-time analysis of the proposed
method.
To this purpose we employ a software implementation in Python~\cite{lutz2013learning},
with SymPy~\cite{meurer2017sympy} as CAS and the linear programming solver from SciPy~\cite{scipy}.
Experiments have been performed
on an Intel Core I3-4340 workstation with dual-channel RAM clocked at 1600 MHz.

Experiments for various examples are reported in Figs.~\ref{fra} and~\ref{frb}.
First, Fig.~\ref{fra} compares the number of codes to explore
after execution of Algorithm~3 (i.e., $|\Psi|$) to the number of
codes to explore in an exhaustive search.
As expected, the number of codes grows as
the minimum probability becomes smaller. However,
it stays significantly below the number of codes necessary
for an exhaustive search, which rapidly reaches values over $10^{80}$.

Execution times 
for the proposed method
are provided in Fig.~\ref{frb}.
Execution times account for both execution of Algorithm~3
and the application of \eqref{eq:closed-form}.
As a comparative reference, Fig.~\ref{frb} also includes 
results for exhaustive search, which are extrapolated
from the proposed method
by 
employing the ratios from Fig.~\ref{fra}
and only considering the time required to calculate \eqref{eq:closed-form}.
From these experimental results,
it seems that the efficiency of the proposed method
comes directly from having less codes to explore.
While pruning costs (i.e., Algorithm~3) may not
be small, they are negligible in relation 
to the large reduction in overall execution time.

\begin{figure}
\begin{center}
\includegraphics[width=\linewidth]{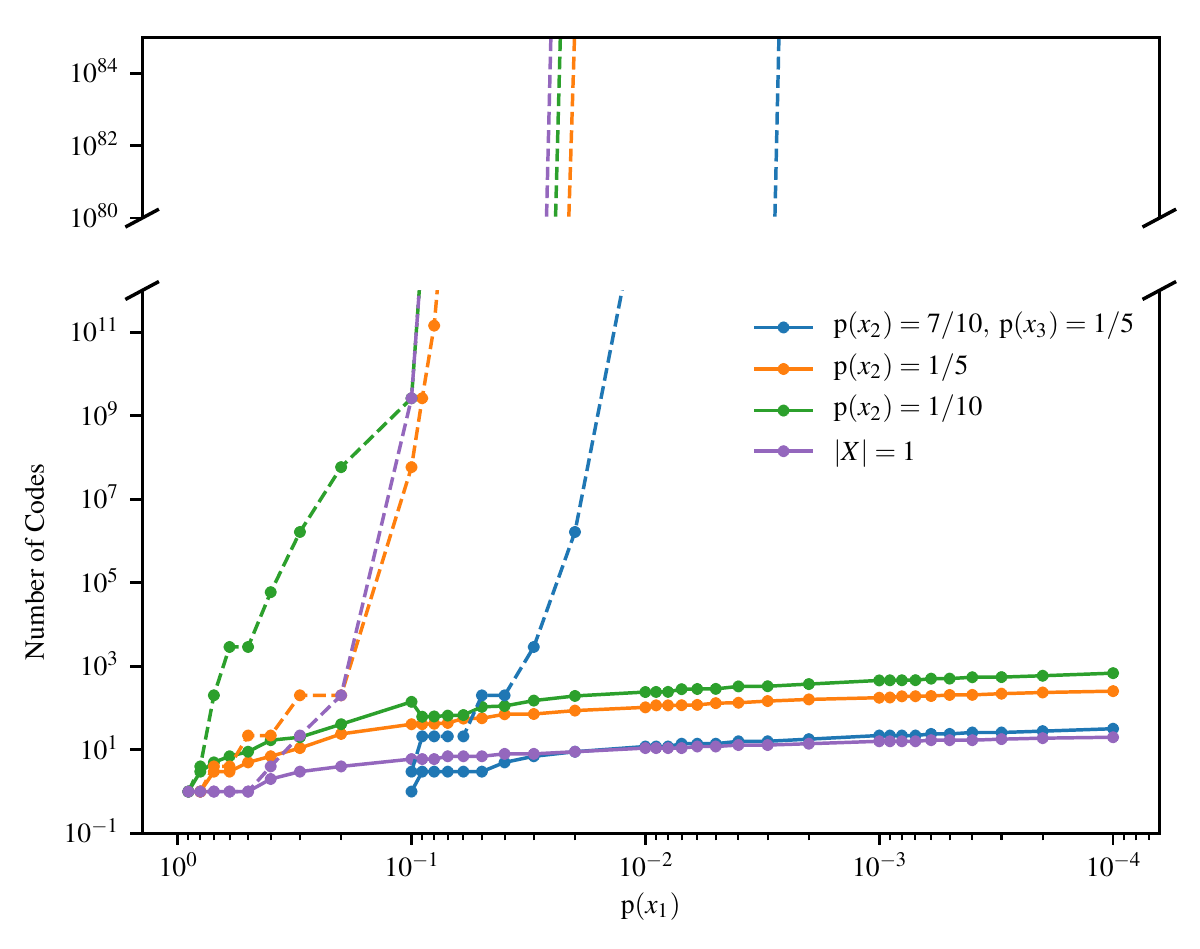}
\end{center}
\caption{Number of codes to be examined 
employing Algorithm~3 (solid curves) and exhaustive search (dashed curves)
for various known probabilities.}
\label{fra}
\end{figure}

\begin{figure}
\begin{center}
\includegraphics[width=\linewidth]{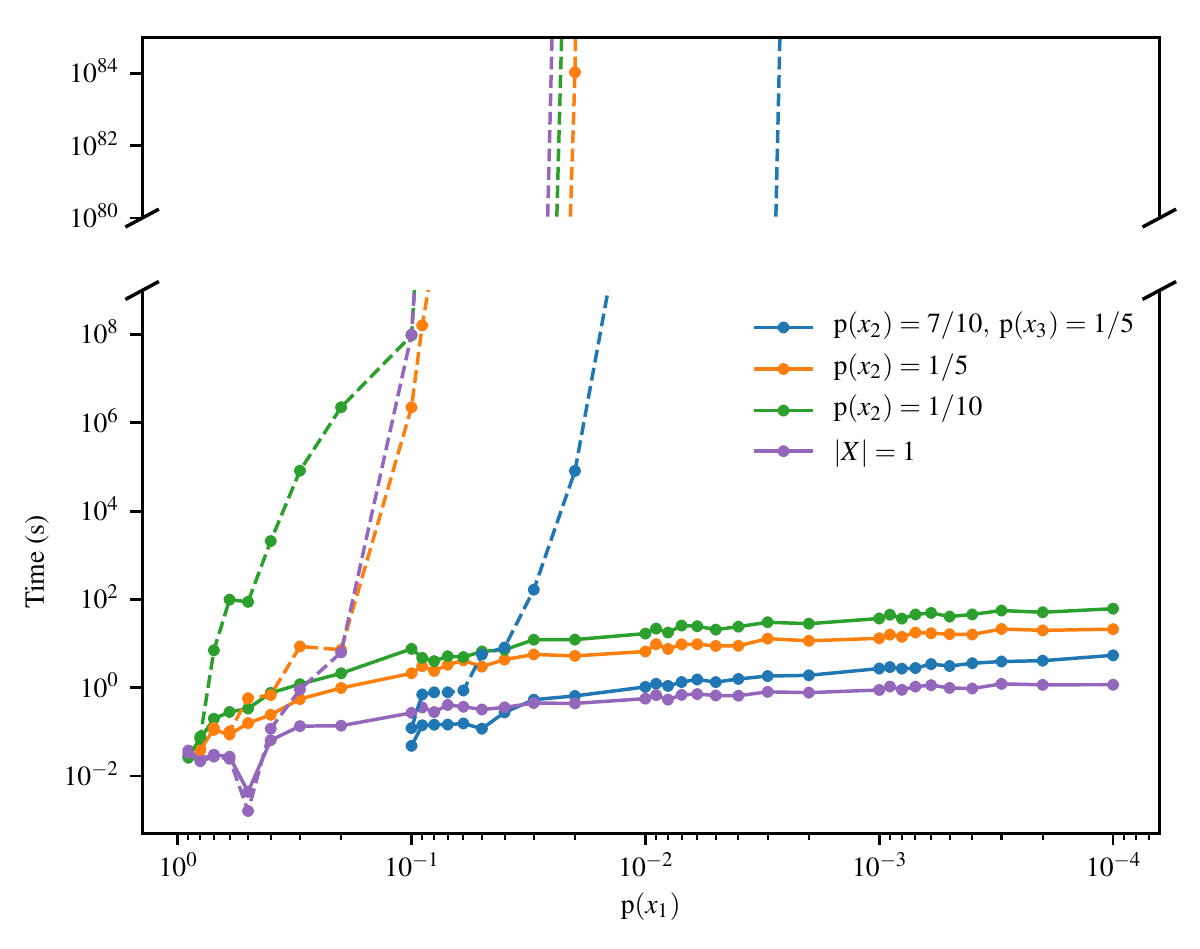}
\end{center}
\caption{Execution time of Algorithm~3 (solid curves) 
and extrapolated time of exhaustive search (dashed curves)
for various known probabilities.}
\label{frb}
\end{figure}


\subsection{Application to V2V Codes}
A few examples on how this bound can be employed
to prune the search space of V2V codes are now presented.
For simplicity, we assume that a V2V code translates variable-length
sequences of input symbols of a memoryless source
into the words of a prefix-free dictionary $W$,
and that those words are then encoded via a Huffman code.
These simplified V2V codes are defined by their dictionary.
For example, for a memoryless source with $\p(a_1) = 0.9$, $\p(a_2) = 0.1$,
one dictionary choice is $W=\{ a_1a_1,a_1a_2,a_2 \}$, which 
can in turn be treated as a three-symbol source 
with probabilities $\p(a_1a_1) = 0.81$, $\p(a_1a_2) = 0.09$, $\p(a_2)=0.1$,
and associated Huffman codewords $1$, $00$, and $01$ (or equivalent).
This defines a V2V code
where $a_1a_1 \mapsto \text{`1'}$, $a_1a_2 \mapsto \text{`00'}$, and 
$a_2 \mapsto \text{`01'}$. 
A search for good codes then becomes a search for good dictionaries.

Given a dictionary $W$, the redundancy for the V2V code that it defines (see~\cite{BDS08}) is 
\begin{equation}
\mathcal{R}_W(A)=\frac{\mathcal{L}_H(W_A) - \mathcal{H}(W_A)}{\sum_{w \in W}{\p(w)\cdot |w|}} = 
\frac{\mathcal{R}_H(W_A)}{\sum_{w \in W}{\p(w)\cdot |w|}} 
\end{equation}
where $|w|$ denotes the length of word $w$ in $W$,
and 
where $W_A$ denotes the source whose symbols are the words in $W$.
The probability of each such symbol is the product of the probabilities of that word.

Given $X \sqsubseteq W_A$ and a maximum word length for $W$,
the minimum redundancy for a V2V code is
\begin{equation}
\mathcal{R}_{\text{min}}(X, L)= 
\frac{\mathcal{R}_{\text{min}}^*(X)}{\sum_{w \in X}{\p(w)\cdot|w|} + L\cdot \sum_{w \notin X}{\p(w)}},
\end{equation}
or
\begin{equation}
\mathcal{R}_{\text{min}}(X, L)= 
\frac{\mathcal{R}_{\text{min}}^*(X)}{L + \sum_{w \in X}{\p(w) (|w| - L)}},
\end{equation}
given that $\sum_{w \notin X}{\p(w)} = 1 - \sum_{w \in X}{\p(w)}$.

In our previous example, we had $\p(a_1a_1) = 0.81$, $\p(a_1a_2) = 0.09$, and $\p(a_2) = 0.1$.
Using the bounds developed here, we can compute 
that any dictionary of $10$ words or less that contains $X=\{a_1a_1, a_1a_2\}$ will yield a
redundancy of no less than 
$\mathcal{R}_{\text{min}}(\{a_1a_1, a_1a_2\}, 10)=\log(2^{-71} \cdot 3^{342}  \cdot 5^{-190}) / 280 \simeq 0.107$ bits per symbol.
This is a minimum encoding overhead of 22.7\% in relation to the source entropy.
If we are targeting a reasonable 1\% overhead, we can clearly discard all dictionaries containing $X$.

Consider another example where a ternary memoryless source
has symbol probabilities of $\p(a_1) = 0.7$, $\p(a_2) = 0.2$ and $\p(a_3) = 0.1$.
The (exact) redundancy for 
$W=\{ a_1 a_1, a_1 a_2, a_1 a_3, a_2a_1, a_2a_2, a_2a_3, a_3a_1, a_3a_2, a_3a_3 \}$
is $\log(2^{73}\cdot 5^{-200} \cdot  7^{140})/200 \simeq 0.008$ bits.
Given that $\mathcal{R}_{\text{min}}(\{a_1, a_3\}, 3)
= \log( 2^8 \cdot 3^{-2} \cdot 5^{-10} \cdot 7^7)/14 \simeq 0.09$ bits,
we know that any dictionary containing words $a_1$ and $a_3$ can not yield 
an optimal V2V code of length $3$ or less.

\section{Conclusions} \label{sect6}

In this paper, we present a tight bound to the minimum redundancy 
achievable by the Huffman code 
when source probabilities may be only partially known.

First, we show how to calculate a bound for alphabets of a given size,
by generating all prefix-free codes that could become a Huffman code, and then,
given the known probabilities, calculating their redundancy through convex optimization. 
This process yields a closed-form number of the minimum redundancy.
The previous bound is further extended to 
alphabets of up to a given size, and then 
generalized to alphabets of any size. This is accomplished by showing
that the last two cases are equivalent under certain conditions.
Moreover, all bounds are tight by construction, as examples lying on each bound are found for each case.

To enable the calculation of the general bound, which may otherwise be
infeasible to calculate, an efficient method is provided to
enumerate all prefix-free codes that could become a Huffman code,
while discarding cases not worth considering. This is achieved 
through early pruning of prefix-free codes which are either
proven more redundant than other codes considered, or 
which are proven by a linear program
to never yield a Huffman code for the given known probabilities.
All obtained results are closed-form numbers
obtainable by means of a CAS, thus preventing any numerical
issue related to floating-point operations.
A numerical solver for linear programs is employed using
a strategy 
that guarantees that numerical issue cannot alter the final result.

In addition, we present examples where we show
the potential of the general bound to aid in the visualization of 
the structure of the minimum redundancy 
achievable by the Huffman code. 
Employing this, we derive a conjecture for a closed-form formula
for $\mathcal{R}^{*}_\text{min}(X)$ when $|X|$ is $2$.

Finally, we would like to remark that the work described here lays the foundation 
to more complex bounds which could incorporate additional restrictions on the probability 
distributions, such as relations between probabilities of multiple symbols, or constraints on their magnitude.

%
\bibliographystyle{IEEEtran}
\bibliography{IEEEabrv,biblio}

\end{document}